\pdfoutput=1
\documentclass[12pt, draftclsnofoot, onecolumn]{IEEEtran}
\usepackage{amsmath}
\usepackage{amsthm}
\usepackage{amsfonts}
\usepackage{amssymb}
\usepackage{mathrsfs}
\usepackage{mathdots}
\usepackage{multirow}
\usepackage{cases}
\usepackage{cite}
\usepackage{graphicx}
\usepackage{epstopdf}
\usepackage{mathrsfs}
\usepackage{subfigure}
\usepackage{epstopdf}
\usepackage{color}
\usepackage{bm}
\usepackage{caption}
\usepackage{booktabs}
\usepackage{stfloats}
\allowdisplaybreaks[1]
\theoremstyle{remark}
\newtheorem{theorem}{\hspace{1em}Theorem}
\newtheorem{lemma}{\hspace{1em}Lemma}

\begin{document}
\bibliographystyle{IEEEtran}

\title{Reliable and Secure Short-Packet Communications
\thanks{C. Feng and H.-M. Wang are with the School of Electronic and Information
Engineering, Xi'an Jiaotong University, Xi'an 710049, China, and also with the Ministry of Education Key Laboratory for Intelligent Networks and Network Security, Xi'an Jiaotong University, Xi'an, 710049, Shaanxi, P. R. China. (Email: {fengxjtu@163.com, xjbswhm@gmail.com}).}
\thanks{H. V. Poor is with the Department of Electrical Engineering, Princeton University, Princeton, NJ 08544, USA. (Email:
poor@princeton.edu).
}
\author{Chen Feng,~Hui-Ming Wang,~\IEEEmembership{Senior Member,~IEEE},~and~H. Vincent Poor,~\IEEEmembership{Life Fellow,~IEEE}
}
}
\maketitle
\vspace{-12mm}
\begin{abstract}
%Exploiting short packet transmission is an effective solution that meets both the reliability and delay requirements in the future internet of things (IoT) services. However, the analysis results of the system secrecy capacity with the classic information theory are based on the condition that the blocklength tends to infinity, which is not suitable for short packet communication. The existing research for short packet communication from the perspective of physical layer security (PLS) is relatively scarce, how to evaluate the system performance with short packets remains to be studied.
Exploiting short packets for communications is one of the key technologies for realizing emerging application scenarios such as massive machine type communications (mMTC) and ultra-reliable low-latency communications (uRLLC).  In this paper, we investigate short-packet communications to provide both reliability and security guarantees simultaneously with an eavesdropper. In particular, an outage probability considering both reliability and secrecy is defined according to the characteristics of short-packet transmission, while the effective throughput in the sense of outage is established as the performance metric. Specifically, a general analytical framework is proposed to approximate the outage probability and effective throughput. Furthermore, closed-form expressions for these quantities are derived for the high signal-to-noise ratio (SNR) regime. Both effective throughput obtained via a general analytical framework and a high-SNR approximation are maximized under an outage-probability constraint by searching for the optimal blocklength.  Numerical results verify the feasibility and accuracy of the proposed analytical framework, and illustrate the influence of the main system parameters on the blocklength and system performance under the outage-probability constraint.

\end{abstract}


\section{Introduction}
\label{Sec:Introduction}
Driven by potential Internet-of-Things (IoT) applications, the fifth-generation (5G) mobile communications system
introduces two novel types of typical scenarios, mMTC and uRLLC. Unlike previous wireless communication systems  pursuing higher transmission rates and higher capacity,  such as 4G long term evolution (LTE) and WiFi, these new service categories put forward higher design requirements for 5G in terms of delay, energy efficiency, reliability, flexibility and connection density \cite{Ji2018Ultra}.
One critical challenge lies in the requirement that 5G must support short-packet transmission, which requires a fundamentally different design approach from that used in current high data rate systems  \cite{durisi2016toward}.

Especially, using short packets for transmission will result in a severe reduction with channel coding gain, making it challenging to ensure communication reliability. In many critical IoT applications, such as autonomous driving, high transmission reliability is a mandatory requirement. Thus, providing highly reliable transmission is one of the key challenges of applications using short-packet communications.

On the other hand, security is also a critical attribute that should be addressed in short-packet communications.  Whether it is the uRLLC scenario for intelligent transportation, industrial control, and telemedicine \cite{Schulz2017Latency}, or the low-cost and low-energy consumption mMTC scenario \cite{Bockelmann2016Massive} with massive numbers of sensors and wearable devices, they are all faced with more severe secrecy challenges than some other communication systems because of the mission-criticality of many IoT applications and because IoT communications often contain private data\cite{Jing2014Security,Shen2020Security}. Taking intelligent transportation as an example, if a message is eavesdropped upon, private information such as user identity or location information may be exposed. Typically the security of communication systems relies on upper-layer encryption, which encounters significant challenges in IoT systems for the following reasons:
\begin{itemize}
\item
Envisioned IoT applications will deploy massive numbers of terminals in many highly dynamic heterogeneous networks, which makes key management very challenging, and is not scalable in terms of key negotiation, distribution, and updating \cite{Morabito2010The,Poor2012Information,2016M2M}.
\end{itemize}
\begin{itemize}
\item
The Resource-constrained IoT nodes have to satisfy practical constraints such as limited energy, so lightweight secrecy protocols are often used to reduce resource requirements at the expense of system reliability, and even security \cite{Zhou2018Security}.
\end{itemize}
\begin{itemize}
\item
A communication mode characterized by numerous sporadic and short packets is used in uRLLC and mMTC applications \cite{durisi2016toward,Poor2012Information,poor2019fundamentals}. The frequent interaction required by traditional security protocols will create considerable additional overhead and consume excessive resources at IoT nodes.
\end{itemize}

Physical layer security (PLS) is a potential candidate for countering the above concerns.  PLS can use the inherent characteristics of  wireless channel to improve the confidentiality of transmissions in the physical layer without relying on a secret key \cite{Bloch2011,liang2009information,Poor2017Wireless}. %Perfect confidentiality can be obtained as the blocklength tends to infinite.
Compared with the security mechanisms based on upper-layer encryption, PLS has lower implementation complexity, more flexible design, and lower dependence on infrastructure \cite{Qi2020PLS}. In recent years, PLS has emerged as a promising complementary and alternative means to ensure wireless transmission security \cite{Nan2015Safeguarding,chen2019,HMwang,HMwangandTXZheng,HMwangandYQ}.

The analysis of reliable and secure performance, as well as corresponding system designs for various communication systems within an infinite blocklength assumption, have been extensively investigated from the perspective of the physical layer \cite{Bloch2011,liang2009information,Poor2017Wireless,Qi2020PLS,Nan2015Safeguarding,chen2019,HMwang,HMwangandTXZheng,HMwangandYQ}. However, the assumption of infinite blocklength is obviously no longer suitable for short-packet transmission in IoT applications. Limited blocklength will introduce performance losses,  for both reliability and security, which is undoubtedly a significant challenge for future applications that require high reliability and security simultaneously, such as in autonomous driving and remote surgery. In addition, using a Shannon information-theoretic framework based on infinite blocklength to analyze the performance of short-packet communication systems may give rise to inaccurate conclusions. Therefore, it is desirable to reconsider the analysis and design of reliable and secure transmissions under finite blocklength from the perspective of the physical layer.

\subsection{Related Work and Motivation}
Over the past decade, the reliable performance of short-packet communication systems has been studied from an information-theoretic perspective. In particular, in \cite{Polyanskiy2010Channel} the maximum achievable channel coding rate under a given fixed blocklength and error probability has been studied in additive white Gaussian noise (AWGN) channels. Unlike in Shannon's asymptotic model, the transmission error probability cannot be arbitrarily low in the case of finite blocklength, while the maximum transmission rate of the system is significantly reduced as well. In subsequent works, bounds on the maximum achievable channel coding rate are analyzed under more specific system models \cite{Yang2013Quasi,yang2014multi,Durisi2016Short} or other channel assumptions \cite{2020saddle}. The basic information-theoretic results in \cite{Polyanskiy2010Channel} have been used to analyze the reliability of various communication systems within finite blocklength, such as two-way relaying \cite{Gu2018Short} and non-orthogonal multiple access (NOMA) system \cite{Yu2018On}, etc. However, none of the above works has taken the confidentiality of wireless communications into consideration.

As for the physical layer security of short-packet communications, some works have studied the secrecy performance with finite blocklength from the information-theoretic perspective. In \cite{Hayashi2006General}, a general achievability bound of the secrecy rate in the presence of an eavesdropper is obtained by using a channel resolution technique. Later research developed improved achievability and converse boundaries for wiretap channels with security metrics \cite{Yassaee2013Non,Tan2013Achievable}. Bounds on the coding rate for the discrete memoryless eavesdropping channel and Gaussian eavesdropping channel are derived under a given blocklength, error probability and information leakage in \cite{Wei2016Finite}, which are tighter than those in \cite{Yassaee2013Non,Tan2013Achievable}. The above issues have been reconsidered in \cite{Wei2017Secrecy}  based on a semi-deterministic eavesdropping channel model, while \cite{Yang2017Wiretap} summarizes and extends the studies \cite{Wei2016Finite,Wei2017Secrecy}. Based on this information-theoretic progress, in \cite{yang2019}, we proposed an analytical framework to approximate the average achievable secrecy throughput of a fading wiretap system with finite blocklength coding. The throughput obtained there is based on the ergodic assumption rather than in the sense of outage.

In fact, for various IoT applications, a significant portion of  traffic consists of bursty short packets transmissions. %Considering these characteristics of short packet communication,
Therefore, performance metrics in the sense of outage is more appropriate. The definition of outage probability for systems with infinite blocklength has been given and applied to system performance analysis\cite{Bloch2008Wireless}.
Nevertheless, this definition cannot be applied for finite blocklength communication systems, as the use of short packets leads to an inevitable decoding error probability and information leakage. Notably, outage-based performance evaluation for secure short-packet communications will confront new challenges due to the complex nature of secure short-packet information theoretic results.

To the best of the authors' knowledge, no outage-based metric considering both reliability and security is available for analyzing the performance of secure short-packet transmission systems, nor has there been a systematic analytical framework for the outage probability and throughput evaluation of such systems. How to design an efficient short-packet transmission scheme to provide both reliability and security guarantees, and what are the impacts of system parameters such as blocklength on the trade-off between transmission rate and reliable-secure performance are still open issues, which motivate this work.

\subsection{Our Contributions}
This paper investigates the effective throughput of short-packet communication systems with eavesdroppers under outage probability constraints. Specifically, the novelty and major contributions of this paper can be summarized as follows.

%We establish a definition of outage probability according to the characteristics of short-packet communications, in order to evaluate performance of the system in a way that can \emph{guarantee reliability and security simultaneously}. Furthermore, we propose an analytical framework to obtain an approximate effective throughput in the finite blocklength transmission. In particular, we obtain the closed-form expressions for the outage probability and effective throughput in the high SNR domain. Based on this, we study the impact of blocklength on the tradeoff between transmission rate and transmission performance within the outage-probability constraint, and further optimize the effective throughput.

\begin{itemize}
\item
Based on the characteristics of short-packet communications, the definitions of outage probability and effective throughput are established, which are used to evaluate the  performance of short-packet communication systems in a way that can \emph{guarantee reliability and security simultaneously}.
\end{itemize}
\begin{itemize}
\item
A general analytical framework is established using a distribution approximation technique that can be applied to analyze and optimize the effective throughput of short-packet communication systems with outage probability constraints. A numerical technique for determining the optimal blocklength is proposed.
\end{itemize}
\begin{itemize}
\item
For the high SNR domain, closed-form expressions for the outage probability and effective throughput for short-packet communication systems are derived, and the optimal blocklength is also obtained  analytically. The impacts of system parameters and outage probability constraints on the optimal blocklength and effective throughput are further analyzed.
\end{itemize}
\begin{itemize}
\item
 Numerical results are provided to verify the feasibility and accuracy of the proposed analytical framework. Several useful insights into the blocklength design for reliable and secure short-packet transmission are revealed. For example, the optimal blocklength and system performance will eventually stop changing with increasing SNR in high SNR domain.
\end{itemize}
\subsection{Organization and Notation}
The rest of this paper is organized as follows: In Section \ref{SystemModel}, we present the system model and the performance metrics for reliable and secure short-packet transmission. The analytical framework for determining the optimal blocklength and effective throughput is given under general system parameters in Section \ref{III}. In Section \ref{IV}, the high SNR case is studied in more detail.
The corresponding numerical simulation results are given in Section \ref{NumericalSimulation}. Section \ref{Conclusion} summarizes the paper and gives the main conclusions.

We use the following notation in this paper: $\bm{a}^T$, $\bm{a}^*$, $|\bm{a}|$, $\left\|\bm{a}\right\|$ denote transpose, conjugate, 1-norm and Euclidean norm of a vector $\bm{a}$, respectively. E$\left\{\cdot\right\}$ and Pr$\left(\cdot\right)$ denote expectation and probability of a random variable, and $\bm{I}_K$ denotes a $K$-dimensional identity matrix. $\mathcal{CN}\left(\mu,\sigma^2\right)$, $Exp\left(\lambda\right)$, and $Gamma\left(k,\theta\right)$ denote a circularly
symmetric complex Gaussian distribution with mean $\mu$ and covariance $\sigma^2$, an exponential distribution with parameter $\lambda$, and a gamma distribution with parameters $k$ and $\theta$, respectively. $Q^{-1}\left(\cdot\right)$ is the inverse of the Gaussian Q-function, $\mathrm{Ei}(x)$ is the exponential integral function, $\psi(x)$ is the Euler psi function, $\zeta(x,q)$ is the Riemann's zeta function and $C$ is Euler's constant.
The ceiling and floor operations are denoted by $\lceil\cdot\rceil$ and $\lfloor\cdot\rfloor$, respectively.

\section{System Model and Problem Description}
\label{SystemModel}
\subsection{System Model}
\label{MISOsystem}
We consider a downlink short-packet transmission system consisting of three nodes as depicted in Fig. \ref{fig1}, in which Alice transmits confidential information to Bob while an eavesdropper Eve attempts to wiretap the ongoing transmission. For IoT applications, the access points usually have abundant resources, whereas the IoT nodes are size and power constrained. Therefore, we assume that Alice and Bob are equipped with $K$ antennas and a single antenna, respectively. In this paper, we consider the case when Eve has a single antenna as well for simplification, and our discussion can be generalized to the multiple antenna case for Eve.

In this paper, quasi-static Rayleigh fading channels are considered, where the channel coefficients remain constant during an entire short-packet communication, and are independently and identically distributed (i.i.d.) among different packets. Accordingly, the channels from Alice to Bob and Eve are represented by $\bm{h}_b\sim\mathcal{CN}\left(\bm{0},\bm{I}_K\right)$ and $\bm{h}_e \sim\mathcal{CN}\left(\bm{0},\bm{I}_K\right)$, respectively. Additionally, AWGN is taken into account for all the channels.

\label{Sec:SystemModelandProblemDescription}
\begin{figure}[!t]
\begin{center}
\includegraphics[width=2.2 in]{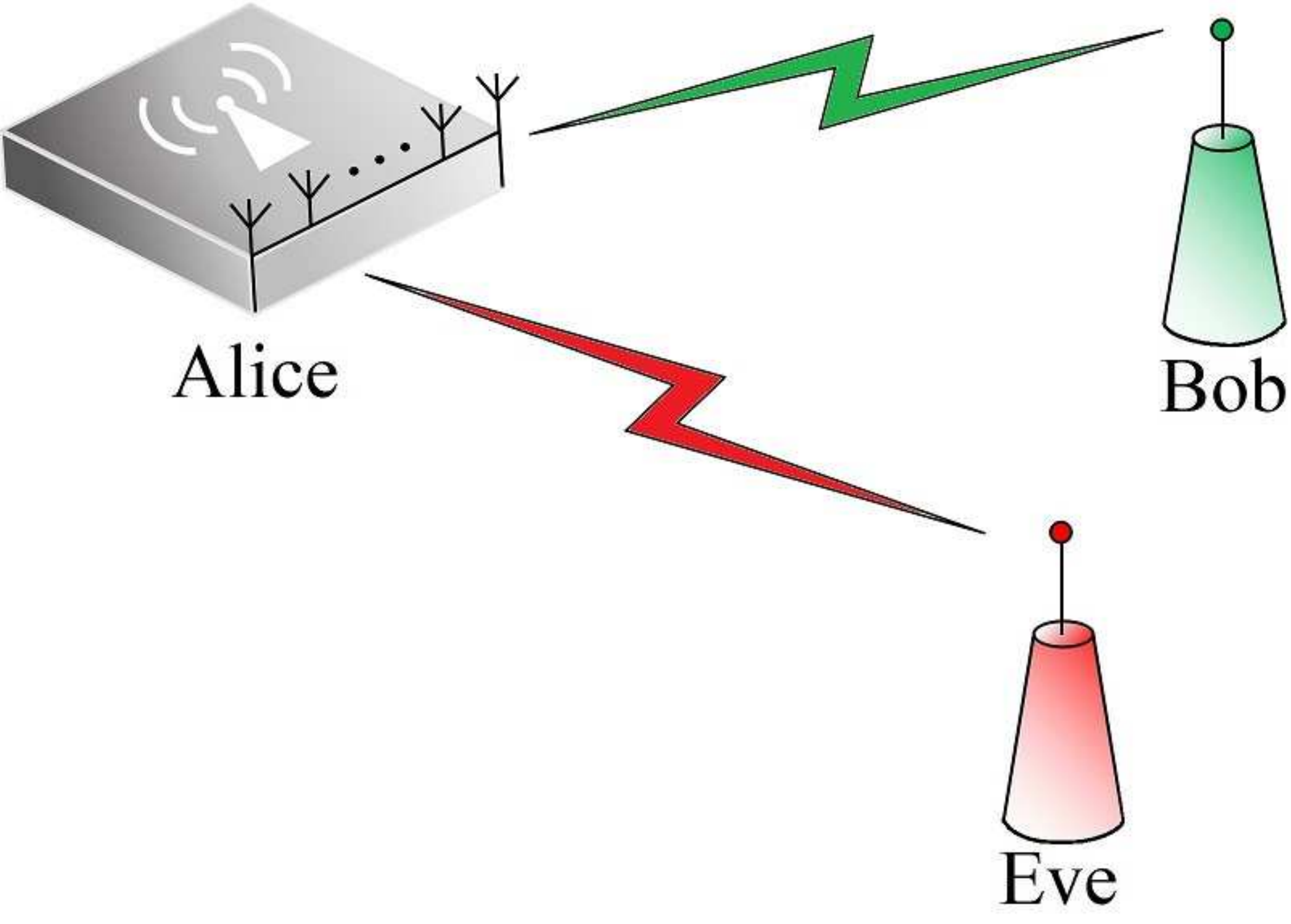}
\vspace{-2mm}
\caption{System Model.}\label{fig1}
\end{center}
\vspace{-10mm}
\end{figure}
Furthermore, it is assumed that $\bm{h}_b$ can be known by Alice through reverse training and channel reciprocity while only the distribution of $\bm{h}_e$ is available at Alice. Accordingly, Alice leverages maximal-ratio transmission (MRT) beamforming scheme with the beamforming vector $w=\bm{h}_b^*/\left\| \bm{h}_b\right\|$. Denoting the information-bearing signal by $\sqrt{P_t}s$ with E$\left\{|s^2|\right\}=1$ and transmit power $P_t$, the received signals at Bob and Eve are given by
\begin{align}
y_b=\bm{h}_b^T\bm{x}+n_b=\sqrt{P_t}\bm{h}_b^T\bm{w}s+n_b,
\label{equation5}
\end{align}
\begin{align}
y_e=\bm{h}_e^T\bm{x}+n_e=\sqrt{P_t}\bm{h}_e^T\bm{w}s+n_e,
\label{equation6}
\end{align}
where $\pmb x\triangleq \sqrt{P_t}\bm{w}s$ is the transmit signal, $n_b\sim\mathcal{CN}\left(0,\sigma_b^2\right)$ and $n_e\sim\mathcal{CN}\left(0,\sigma_e^2\right)$ accounts for the AWGN at Bob and Eve, respectively. The corresponding SNRs can be formulated as
\begin{align}
\gamma_b=\frac{P_t\left\|\bm{h}_b^T\bm{w}\right\|^2}{\sigma_b^2}=\overline{\gamma}_b\left\| \bm{h}_b\right\|^2,
\label{equation7}
\end{align}
\begin{align}
\gamma_e=\frac{P_t\left\|\bm{h}_e^T\bm{w}\right\|^2}{\sigma_e^2}=\overline{\gamma}_e\left\|\bm{h}_e^T\frac{\bm{h}_b^*}{\left\| \bm{h}_b\right\|}\right\|^2,
\label{equation8}
\end{align}
where $\overline{\gamma}_b\triangleq P_t /\sigma_b^2$ ,$\overline{\gamma}_e\triangleq P_t /\sigma_e^2$. From (\ref{equation7}) and (\ref{equation8}), we immediately have $\gamma_b\sim Gamma\left(K,\overline{\gamma}_b\right)$ and $\gamma_e\sim Exp\left(1/\overline{\gamma}_e\right)$. In this paper, we assume that the noise power is equal at Bob and Eve, i.e. $\sigma_b^2=\sigma_e^2$, then $\overline{\gamma}_b=\overline{\gamma}_e=\overline{\gamma}$.

\subsection{Finite Blocklength vs. Infinite Blocklength}
Traditional wireless communication systems usually take ``capacity'' as the critical performance metric, which has a standard assumption that channel coding with a sufficient long  (infinite) blocklength should be used to achieve the capacity. Similarly, secrecy capacity, the maximum transmission rate that legitimate transmit-receive ends can achieve when they communicate in a completely secure and reliable way\footnote{In the sense that the error probability and the information leakage can be made
as small as desired, as long as the transmission rate is below
the secrecy capacity.}, is utilized to evaluate the secrecy performance of a wiretap channel. In \cite{Bloch2008Wireless}, it has been strictly proved that in a quasi-static  fading wiretap channel, a positive secrecy rate exists when Bob's channel capacity $C_b$ is larger than Eve's channel capacity $C_e$. In this case, the closed-form expression for the secrecy capacity can be represented as
\begin{align}
C_s= C_b-C_e =\log_2\left(1+\gamma_b\right)-\log_2\left(1+\gamma_e\right), \label{Cs}
\end{align}
where  $\gamma_b$, $\gamma_e$ denote the SNRs at Bob and Eve respectively.
%In a general wiretap channel model rather than Guassian wiretap channel, the accuracy secrecy capacity is not easy to obtain because in general a very difficult non-convex optimization problem is involved. In the case, we have
%\begin{align}
%C_s \geq R_s = R_b-R_e,
%\end{align}
%where $R_s$ is a achievable secrecy rate, and $R_b$ and $R_e$ are achievable rates of Bob's and Eve's channels.
This rate can be achieved by the well-known Wyner's wiretap coding scheme for infinite blocklength.
The scheme involves two rates, namely, codeword rate $R_b$, and secrecy rate $R_s$ separately. The rate difference, i.e.,
\begin{equation}
R_e = R_b -R_s, \label{Rs_}
\end{equation}
called rate redundancy, reflects the rate cost of securing the message transmission against eavesdropping\cite{Wyner1975,Zhou2011Rethinking}.  For systems with infinite blocklength, both $R_b$ and $R_e$ can be designed separately to achieve a possible secrecy rate $R_s$ that is close to $C_s$.

However, secrecy capacity based on infinite blocklength are no longer suitable for the analysis of short-packet communications in uRLLC and mMTC scenarios. There is considerable interest in investigating the secrecy capacity incurred by
coding at a finite blocklength. %According to \cite{Polyanskiy2010Channel},the maximum achievable channel coding rate for a given blocklength $N$ with a constraint on the decoding error probability of $\epsilon$ at the receiver is approximated as follows:
%\begin{align}
%R \approx
%\log_2\left(1+\gamma\right)-\sqrt{\frac{1-\left(1+\gamma\right)^{-2}}{N}}\frac{Q^{-1}\left(\epsilon\right)}{\ln2},
%\label{equation1}
%\end{align}
%where $\gamma$ denotes the SNR at the legitimate receiver Bob, and
%$Q^{-1}\left(\right)$ is the inverse of the Gaussian Q-function. From (\ref{equation1}), it can be observed that a penalty term is introduced compared with the infinite blocklength case, which is proportional to $\frac{1}{\sqrt{N}}$. If $N$ approaches infinity, the penalty diminishes and the maximum achievable channel coding rate given in (\ref{equation1}) asymptotically coincides with the channel capacity.
%
%Similarly,
The maximal secrecy transmission rate $R^*(N,\epsilon,\delta)$ for a given blocklength $N$, a decoding error probability $\epsilon$, and an information leakage $\delta$, over wiretap channel has been investigated very recently in \cite{Yang2017Wiretap}. Specifically, the achievability and converse bounds on the second-order coding rate of quasi-static Rayleigh fading wiretap channel can be represented as
\begin{align}
C_s-\sqrt{\frac{V_1}{N}}\frac{Q^{-1}\left(\epsilon\right)}{\ln2}-\sqrt{\frac{V_2}{N}}\frac{Q^{-1}\left(\delta\right)}{\ln2}\mathop {<}_{\approx} R^*(N,\epsilon,\delta)\mathop {<}_{\approx}C_s-\sqrt{\frac{V_3}{N}}\frac{Q^{-1}\left(\epsilon+\delta\right)}{\ln2},
\label{equation_up_and_down}
\end{align}
where $V_1\triangleq 1-\left(1+\gamma_b\right)^{-2}$, $V_2\triangleq 1-\left(1+\gamma_e\right)^{-2}$, and $V_3\triangleq V_1+V_2-2\frac{1+\gamma_e}{1+\gamma_b}V_2$ are constants that depend on the stochastic variations of the legitimate's and the eavesdropper's channel\footnote{The equation in [\cite{Yang2017Wiretap}, eq. (115-118)] is adapted here for the complex-valued channel by noting that the equivalent blocklength gets doubled compared with the real-valued one with the same SNR.}. For a packet with blocklength $N$, if the probability of decoding error at Bob is not larger than $\epsilon$ while the information leakage does not exceed $\delta$, the maximum transmission rate of Rayleigh fading wiretap channel should be within the range shown in (\ref{equation_up_and_down}). Compared with the infinite blocklength case, the finite blocklength $N$ leads to a loss of reliability and secrecy, and the perfect reliability and secrecy of short-packet transmission cannot be guaranteed. It is worth noting that the upper and lower bounds of $R^*(N,\epsilon,\delta)$ will converge to $C_s$ together when $N$ approaches infinity.

The lower bound of $R^*(N,\epsilon,\delta)$ in (\ref{equation_up_and_down}) is achievable under the constraints of reliability and secrecy, which will be exploited in the following analysis to approximate the upper bound of outage probability. This achievable secrecy rate is
\begin{align}
\bar R_s=C_s-\sqrt{\frac{V_b}{N}}\frac{Q^{-1}\left(\epsilon\right)}{\ln2}-\sqrt{\frac{V_e}{N}}\frac{Q^{-1}\left(\delta\right)}{\ln2},
\label{equation2}
\end{align}
where $V_i\triangleq 1-\left(1+\gamma_i\right)^{-2}, i\in\{b, e\}$. The eq.
(\ref{equation2}) implies that there exists a wiretap code with blocklength $N$ and coding rate $\bar R_s$, such that the probability of decoding error at Bob is not larger than $\epsilon$ while the information leakage does not exceed $\delta$.
Reformulating (\ref{equation2}), we have
\begin{equation}
\bar R_s= \Re_b-\Re_e, \label{Rs}
\end{equation}
with
\begin{align}
&\Re_b\triangleq\log_2\left(1+\gamma_b\right)-P\sqrt{V_b},\label{equationRb}\\
&\Re_e\triangleq\log_2\left(1+\gamma_e\right)+Q\sqrt{V_e}.\label{equationRe}
\end{align}
where $P=\frac{Q^{-1}\left(\epsilon\right)}{\ln2\sqrt{N}}$, $Q\triangleq\frac{Q^{-1}\left(\delta\right)}{\ln2\sqrt{N}}$. We see that (\ref{Rs}) has a similar form as (\ref{Cs}) and (\ref{Rs_}). Furthermore, $\Re_b$ is only related to the legitimate channel $\gamma_b$ and the reliable constraint $\epsilon$, whereas $\Re_e$ is related to the wiretap channel
 $\gamma_e$ and the secure constraint $\delta$.
Before proceeding, we have to make some important observations on the  differences %of secrecy rates
between finite blocklength and infinite blocklength.

We should emphasize that (9) merely expresses $R_s$ as two terms in a mathematical form. The $\Re_b$ in (\ref{equationRb}) and $\Re_e$ in (\ref{equationRe})  represent the components related to reliable transmission constraint and information leakage constraint in $\bar R_s$, separately, which are different from the main channel rate and eavesdropping channel rate in the typical transmission system with infinite blocklength.

When infinite blocklength coding scheme is adopted, an accurate and unique secrecy capacity $C_s$ can be obtained for a given channel realization. With Wyner's wiretap coding scheme, the codeword rate and rate redundancy can be designed and optimized \emph{separately} to achieve a possible secrecy rate $R_s$ that close to $C_s$, which has been extensively studied in \cite{2013CodingforSecrecy,Error-ControlCoding}. However, for short-packet transmission with finite blocklength, Wyner's wiretap coding scheme may not hold anymore, which means that there is no separable codeword rate and redundancy rate.
Instead, we have to design a coding rate $R_0$ to approach the achievable secrecy rate $\bar R_s$ in (\ref{Rs}) directly, although $\bar R_s$ can be decomposed into the forms of $\Re_b$ and $\Re_e$ as shown in (\ref{equationRb}) and (\ref{equationRe}), respectively. This conclusion is important, as we will see in our following derivations.

Finally, for systems within infinite blocklength, when the capacity of main channel is larger than that of eavesdropping channel, there always exists a coding scheme that can guarantee a \emph{perfectly} reliable and secure transmission (decoding error and information leakage approach zero), simultaneously. On the contrary, for short-packet communications, due to the penalty introduced by the limited blocklength, perfect reliability and security cannot be guaranteed anymore. We can only analyze the system performance under a certain reliable constraint (decoding error at Bob is not larger than $\epsilon$) and a secure constraint (the information leakage does not exceed $\delta$). This observation has a significate impact on the definition of outage event and the analysis on the outage probability, as will be discussed in the following subsection.

\subsection{Outage Probability for Short-Packet Communications}
In this subsection, we give the definition of outage event and analyze the outage
probability for systems with finite blocklength. %Traditional secrecy outage for infinite blocklength encoding refers to a situation where perfect reliability and confidentiality  cannot be guaranteed \cite{Barros2006Secrecy,Bloch2008Wireless}.
According to the discussion in Section II-B, perfect reliability and security cannot be guaranteed, thus we have to measure the system performance under the restrictions of decoding error probability $\epsilon$ and information leakage probability $\delta$. Based on this observation, we define the outage event as the case that the transmission with current rate $R_0$ and blocklength $N$ violates either a preset reliable constraint $\bar \epsilon$ or a preset secure constraint $\bar \delta$.

Specifically, assume that $B$ bits confidential information is transmitted during $N$ channel uses. In other words, the confidential information is transmitted by Alice at a rate of $R_0=B/N$. After modulation, the confidential signal will go through the legitimate and wiretap channels. For each realization of the fading coefficients $\{\bm{h}_b, \bm{h}_e\}$, the quasi-static fading channels can be viewed as AWGN channels with fixed channel gains, and the current achievable secrecy rate $\bar R_s$ under the preset constraints of reliability and security $\{\bar \epsilon, \bar \delta\}$ can be calculated by substituting the instantaneous SNRs into (\ref{equation2}). However, since Alice does not know the exact $\bar R_s$ because of the unknown instantaneous SNR $\gamma_e$, we will encounter two cases:

In case of $R_0\leq \bar R_s$, which means that the current coding rate is less than or equal to the achievable secrecy rate, both the preset constraints of reliability $\bar \epsilon$ and security $\bar \delta$ can be satisfied.

In case of $R_0> \bar R_s$, which means that the current coding rate exceeds the achievable secrecy rate, either the preset reliable constraint $\bar \epsilon$ or the preset secure constraint $\bar \delta$ is invalid. It is reasonable to define such an event as an \emph{outage event}.
Due to the time variation of the channels and consequently the randomness of $\bar R_s$, for any coding rate $R_0$, the outage event is a probability event with the outage probability defined as
\begin{align}
p_{out}(\bar \epsilon, \bar \delta )=Pr\left(\bar R_s\leq R_0\right).
\label{equation3}
\end{align}
Recalling $\bar R_s$ in (\ref{Rs}), the outage probability can be interpreted as the probability that the main channel gain is too small or the eavesdropping channel gain is too large to allow for a communication satisfying the reliable and secure constraints simultaneously with current rate $R_0$.

\subsection{Effective Throughput}
To measure the transmission performance, we establish the concept of effective throughput $T$, which measures the average effectively received information bits per channel use subject to an outage constraint $p_{out}\leq\zeta$, where $\zeta\in[0,1]$ is a pre-established threshold of outage probability\footnote{Generally, the outage probability should be no more than 0.5, i.e. $\zeta\leq 0.5$.
}.
Specifically, the \emph{effective} transmission refers to the transmission that satisfies both the reliable and secure constraints, which has a probability of $(1-p_{out})$. Since each short packet of blocklength $N$ contains $B$ information
bits in total, the effective throughput $T$ can be formulated as
\begin{align}
T=R_0\left(1-p_{out}\right)=\frac{B}{N}\left(1-p_{out}\right)\qquad \mathrm {s.t.} \quad p_{out}\leq\zeta.
\label{equation4}
\end{align}

For each short-packet transmission, the number of confidential information bits $B$ should be fixed, which is determined by the specific requirement in a certain %uRLLC or mMTC
application scenario. After secrecy coding,  the number of coded bits is $N$, which actually captures the corresponding
physical-layer transmission latency measured in channel
uses and should be carefully designed.

It can be observed from (\ref{equation2}) that a larger $N$ can increase the achievable secrecy rate $R_s$, while the transmission rate $R_0$ decreases as $N$ increases. Accordingly, a larger $N$ will result in a smaller $p_{out}$ according to (\ref{equation3}). Obviously, the choice of $N$ actually implies a tradeoff between transmission rate and reliable-secure performance. Therefore, it is important to investigate the impact of $N$ on the effective throughput defined in (\ref{equation4}). The value of $N$ should be carefully chosen so as to maximize the effective throughput under the upper bound constraint on the blocklength $N\le N_G$.
Accordingly, the optimization problem is
\begin{subequations}
	\begin{align}
	&\mathop {\mathrm {maximize}}_{N}~T \label{equation13a}\\
	&\mathrm {s.t.}~ p_{out}\leq\zeta, \label{equation13b}\\
	&\hphantom { \mathrm {s.t.}~}N\leq N_G \text{ with } N\in \mathbb {N}_{+}. \label{equation13c}
	\end{align}
	\label{equation22}
\end{subequations}
According to the definition of $p_{out}$  in (\ref{equation3}), the constraint (\ref{equation13b}) is a function of $R_0$, which is also related to $N$. Therefore, the objective function and the constraints are coupled. To solve the problem, the first task is to derive an explicit expression of the outage probability $p_{out}$. However, we will show that it is a non-trivial task.
In the subsequent section, we will propose a distribution approximation technique to address this issue and then solve the optimization problem.

\section{A general analysis and optimization framework}
\label{III}
For any given $R_0$, a prerequisite for acquiring the outage probability in (\ref{equation3}) is to obtain the explicit expression of the probability density function (PDF) of the achievable secrecy rate $\bar R_s$, denoted by $f_{\bar R_s}(r)$. However, according to (\ref{Rs})-(\ref{equationRe}), we see that for any preset $\bar \epsilon$ and $\bar \delta$,\footnote{The preset decoding error probability $\bar \epsilon$, and the information leakage $\bar \delta$ should be no more than 0.5, i.e. $\bar\epsilon\leq 0.5$, $\bar\delta\leq0.5$.}  $\Re_b$ and $\Re_e$ are both complicated functions of $\gamma_b$ and $\gamma_e$ ($V_b$ and $V_e$ are also functions of $\gamma_b$ and $\gamma_e$, respectively), making it challenging to derive $f_{\bar  R_s}\left(r\right)$ from the distributions of $\gamma_b$ and $\gamma_e$ directly.
%There are two ways to solve this problem. On the one hand, the expression of $R_s$ about $\gamma_b$ and $\gamma_e$ can be simplified under certain conditions, e.g. with the assumption of %high SNR, so that the distribution of $R_s$ can be obtained by solving the distribution of the function about $\gamma_b$ and $\gamma_e$. On the other hand,

In this section, we propose to use a distributed approximation technique to address this issue. Specifically, a distribution model with a known expression and some unknown parameters is used to fit the original distribution of $\bar R_s$. Calculate the unknown parameters on the basis of some fitting criteria to obtain an approximate version of $f_{\bar R_s}(r)$. With the aid of the approximate $f_{\bar R_s}\left(r\right)$, expressions for the outage probability $p_{out}$ and the effective throughput $T$ can be obtained. %by solving the optimization problem (\ref{equation22})., and the solution of the optimization problem can be numerically calculated through an one-dimensional search.

\subsection{Distribution Approximation for $f_{\bar R_s}(r)$}
%In addition to the above method, it is also feasible to use the known distribution model for fitting the distribution of $R_s$ directly.
\begin{table}
	\centering
	\caption{Distribution approximation for $f_{\bar R_s}(r)$}
	\fbox{%
		\parbox{\textwidth}{%
			\begin{itemize}
				\item [1.]
				Select a suitable target approximate function model according to the distribution of $\bar R_s$. The model should fit the original distribution well and have as few parameters as possible.
				\item [2.]
				Calculate the moments of $\Re_b$ and $\Re_e$ separately based on distributions of $\gamma_b$ and $\gamma_e$, and then the moment of $\bar R_s$ can be easily deduced since $\Re_b$ and $\Re_e$ are independent of each other.
				\item [3.]
				Construct equations between the moments and distribution parameters of the objective approximate function model, and then using the moments as links to represent the unknown parameters of the objective function with the moments of $\bar R_s$.
			\end{itemize}
		}%
	}
	\vspace{-8mm}
	\label{Table1}
\end{table}
Using the distribution approximation technique to obtain an approximate version of an unknown distribution consists of two key steps: one is to determine the approximate distribution model, and the other is to select an estimation criterion for determining distribution parameters. Specifically, the approximate distribution model should meet several basic requirements, including being able to fit the original distribution well and involving less unknown parameters that are easy to be obtained. As for the estimation criteria of distribution parameters, various estimation criteria such as least squares, maximum likelihood and moments methods have been exploited for the unknown distribution parameters estimation in \cite{Tsai2004Estimation, Jim2014Estimation}
%, which have different approximate performance and complexities
. In this paper, we choose the moments method for obtaining the approximate expression of $f_{\bar R_s}(r)$, whose analytical framework is shown in Table \ref{Table1}. Especially, we choose Gaussian distribution as the target approximate distribution model. Reasons for choosing Gaussian distribution are two folds.
On the one hand, only two parameters (mean and variance) are involved, which leads to a relative low complexity of determining the unknown parameters. More critically, as we will see, the relationship between the model parameters (mean and variance) and the moment characteristics of $\Re_b$ and $\Re_e$ is very simple, and the computational complexity for the corresponding parametric equations is quite low. On the other hand, the approximation performance is excellent in a large range of SNR, as will be shown by the simulation results. In a word, Gaussian approximate distribution model involves a satisfying trade-off between approximate accuracy and complexity\footnote{It should be particularly noted that other distribution models with better adaptability can also be used for approximation here, more accurate approximation may be obtained compared with Gaussian models. For example,  Weibull distribution or generalized normal distribution can be used to obtain similar or even better approximations, which can provide more universal distribution fittings. However, these distributions involve more parameters that require an equivalent number of moments. What's more, the equations established between the parameters and the moment features are very difficult to solve, if not impossible.}.

%\label{distributionapproximation}
%\begin{figure*}[t]
%\begin{center}
%\subfigure[$\overline\gamma=0dB$] {
%\begin{minipage}[b]{0.3\textwidth}
%%\includegraphics[width=2.3 in]{New1.eps}\\
%\includegraphics[width=2.3 in]{Rs_0dB.eps}
%\end{minipage}
%}
%\hfill
%\subfigure[$\overline\gamma=10dB$] {
%\begin{minipage}[b]{0.3\textwidth}
%%\includegraphics[width=2.3 in]{New3.eps}\\
%\includegraphics[width=2.3 in]{Rs_10dB.eps}
%\end{minipage}
%}
%\hfill
%\subfigure[$\overline\gamma=20dB$] {
%\begin{minipage}[b]{0.3\textwidth}
%%\includegraphics[width=2.3 in]{New5.eps}\\
%\includegraphics[width=2.3 in]{Rs_20dB.eps}
%\end{minipage}
%}
%\caption{\small{PDF and CDF of $R_s$ and their approximations under different SNRs with $N=200$, $\epsilon=\delta=10^{-3}$, $K=4$}}
%\label{fig2}
%\end{center}
%\end{figure*}

Following the analytical framework in Table \ref{Table1}, the mean and variance of $\bar R_s$ should be deduced with the PDFs %$f_{\gamma_b}= Gamma\left(K,\overline\gamma\right)$ and $f_{\gamma_e} = Exp\left(1/\overline\gamma\right)$. Let $x$ and $y$ stand for the value of $\gamma_b$ and $\gamma_e$, thenThe PDFs of $\gamma_b$ and $\gamma_e$ are
$f_{\gamma_b}\left(x\right)=\frac{1}{\overline\gamma^K\Gamma\left(K\right)}x^{K-1}\mathrm{e}^{-\frac{x}{\overline\gamma}}$ and $f_{\gamma_e}\left(y\right)=\frac{1}{\overline\gamma}\mathrm{e}^{-\frac{y}{\overline\gamma}}$. Then the mean and the second-order origin moment of $\Re_b$ can be respectively represented as
\begin{align}
m_{\Re_b}&=\int_{0}^{\infty} \left(C_b\left(x\right)-P\sqrt{V_b\left(x\right)}\right)f_{\gamma_b}\left(x\right)\, dx=p_0-p_1P,
\label{equation9}\\
E\left\{\Re_b^2\right\}&=\int_{0}^{\infty} \left(C_b\left(x\right)-P\sqrt{V_b\left(x\right)}\right)^2f_{\gamma_b}\left(x\right)\, dx=p_2-p_3P+p_4P^2,
\label{equation10}
\end{align}
where $p_0\triangleq\int_{0}^{\infty} C_b\left(x\right)f_{\gamma_b}\left(x\right)\,dx$, $p_1\triangleq\int_{0}^{\infty}\sqrt{V_b\left(x\right)}f_{\gamma_b}\left(x\right)\,dx$, $p_2\triangleq\int_{0}^{\infty}C_b^2\left(x\right)f_{\gamma_b}\left(x\right)\,dx$,
$p_3\triangleq\int_{0}^{\infty}2C_b\left(x\right)\sqrt{V_b\left(x\right)}f_{\gamma_b}\left(x\right)\,dx$, and
$p_4\triangleq\int_{0}^{\infty}V_b\left(x\right)f_{\gamma_b}\left(x\right)\,dx$.
Then the variance of $\Re_b$ is
\begin{align}
\sigma_{\Re_b}^2=(p_4-p_1^2)P^2+\left(2p_0p_1-p_3\right)P+p_2-p_0^2.
\label{equation11}
\end{align}
All the parameters are  obviously positive and can be numerically calculated.
Similarly, the mean and variance of $\Re_e$ are represented by
\begin{align}
m_{\Re_e}=q_0+q_1Q,\qquad \sigma_{\Re_e}^2=(q_4-q_1^2)Q^2+\left(q_3-2q_0q_1\right)Q+q_2-q_0^2 \label{equation12}
\end{align}
where $q_0\triangleq\int_{0}^{\infty} C_e\left(y\right)f_{\gamma_e}\left(y\right)\,dy$, $q_1\triangleq\int_{0}^{\infty}\sqrt{V_e\left(y\right)}f_{\gamma_e}\left(y\right)\,dy$, $q_2\triangleq\int_{0}^{\infty}C_e^2\left(y\right)f_{\gamma_e}\left(y\right)\,dy$,
$q_3\triangleq\int_{0}^{\infty}2C_e\left(y\right)\sqrt{V_e\left(y\right)}f_{\gamma_e}\left(y\right)\,dy$, and
$q_4\triangleq\int_{0}^{\infty}V_e\left(y\right)f_{\gamma_e}\left(y\right)\,dy$.

For the reason that $\gamma_b$ and $\gamma_e$ are independent, we can obtain the approximate result of $\bar R_s$ with the following theorem.
\begin{theorem}
	The distribution of $\bar R_s$ can be approximated as normal distribution, and the CDF $\tilde F_{\bar R_s}(r)$ can be formulated as
	\begin{align}
	\tilde F_{\bar R_s}(r)=\Phi\left(\frac{r-m_{\bar R_s}}{\sigma_{\bar R_s}}\right),
	\label{equation14}
	\end{align}
	with $m_{\bar R_s}=m_{\Re_b}-m_{\Re_e}$ and $\sigma_{\bar R_s}^2 =\sigma_{\Re_b}^2+\sigma_{\Re_e}^2$,
%	\begin{align}
%	m_{\bar R_s}&=m_{\bar R_b}-m_{\bar R_e},\label{equation15}\\
%	\sigma_{\bar R_s}^2 &=\sigma_{\bar R_b}^2+\sigma_{\bar R_e}^2,\label{equation16}
%	\end{align}
where $\Phi(x)\triangleq \frac{1}{\sqrt{2\pi}}\int_{-\infty}^{x}\mathrm{e}^{-\frac{t^2}{2}}dt$.
\label{theorem1}
\end{theorem}
\begin{proof}
	According to the definitions in (\ref{Rs})-(\ref{equationRe}) and the independence between $\gamma_b$ and $\gamma_e$, we immediately have the results.
\end{proof}

\begin{figure}[!t]
	\begin{center}
		\includegraphics[width=3 in]{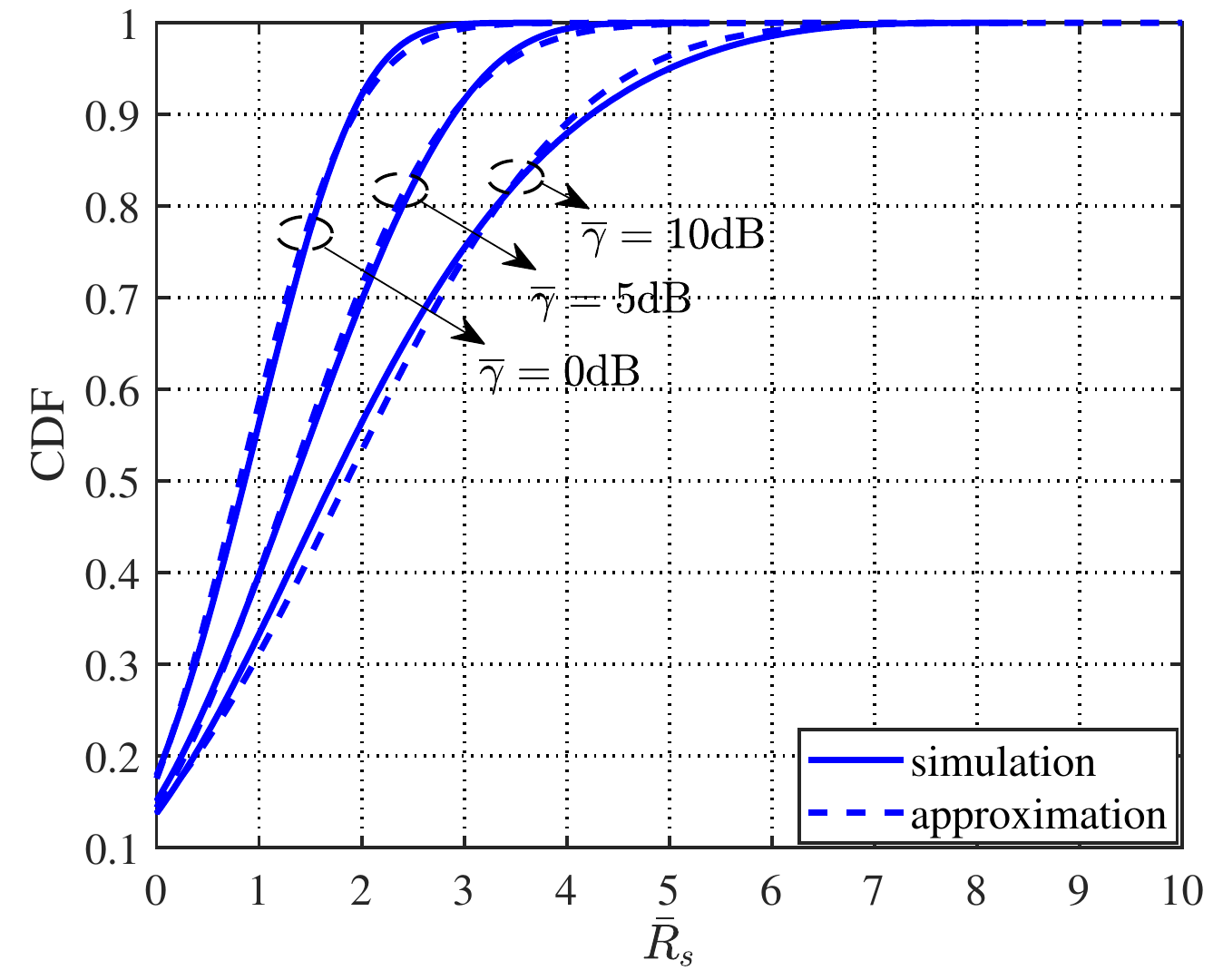}
	\end{center}
\vspace{-5mm}
	\caption{\small{CDFs of $\bar R_s$ and their approximations under different $\overline\gamma$ with $N=200$, $\bar\epsilon=\bar\delta=10^{-3}$, $K=4$.}}
	\label{fig2}
	\vspace{-10mm}
\end{figure}

To facilitate subsequent analysis, based on (\ref{equation9})-(\ref{equation12}), we reformulate $m_{\bar R_s}$ and $\sigma_{\bar R_s}^2$as
\begin{align}
m_{\bar R_s}=-\frac{\mu_1}{\sqrt{N}}+\mu_0,\qquad \sigma_{\bar R_s}^2=\frac{\mu_4}{N}+\frac{\mu_3}{\sqrt{N}}+\mu_2, \nonumber
\end{align}
where $\mu_0\triangleq p_0-q_0$, $\mu_1\triangleq \frac{p_1Q^{-1}(\bar \epsilon)+q_1Q^{-1}(\bar \delta)}{\ln2}$, $\mu_2\triangleq p_2-p_0^2+q_2-q_0^2$, $\mu_3\triangleq \frac{(2p_0p_1-p_3)Q^{-1}(\bar \epsilon)-(2q_0q_1-q_3)Q^{-1}(\bar \delta)}{\ln2}$, and $\mu_4\triangleq \frac{\left(p_4-p_1^2\right)\left(Q^{-1}(\bar \epsilon)\right)^2+\left(q_4-q_1^2\right)\left(Q^{-1}(\bar \delta)\right)^2}{(\ln2)^2}$.
The approximate performance is shown in Fig. \ref{fig2}. As the curves show, the CDF curves obtained by the approximate model in Theorem \ref{theorem1} fit the true CDF of $\bar R_s$ well in a large range of SNRs. Under the currently selected Gaussian approximation model, the approximate error tends to be larger as the average SNR $\overline\gamma$ increases.

\subsection{Effective Throughput Optimization }
With the approximate CDF of $\bar R_s$, the outage probability can be represented as
\begin{align}
p_{out}=\mathrm{\Phi}\left(\frac{R_0-m_{\bar R_s}}{\sigma_{\bar R_s}}\right).
\label{equation19}
\end{align}
Then the optimization problem in (\ref{equation22}) can be reformulated as
%\begin{align}
%	&. \label{approxT} \\
%	&\mathrm {s.t.}~ \mathrm{\Phi}\left(\frac{ \frac{B}{N}-m_{\bar R_s}}{\sigma_{\bar R_s}}\right) \leq\zeta, \nonumber\\
%	&\hphantom { \mathrm {s.t.}~}N\leq N_G, \nonumber\\
%    &\hphantom { \mathrm {s.t.}~}N\in \mathbb {N}_{+}. \nonumber
%\end{align}
\begin{subequations}
	\begin{align}
	&\mathop {\mathrm {maximize}}_{N}~ T=\frac{B}{N}\mathrm{\Phi}\left(\frac{m_{\bar R_s}-\frac{B}{N}}{\sigma_{\bar R_s}}\right) \label{approxTa}\\
	&\mathrm {s.t.}~ p_{out}\leq\zeta, \label{approxTb}\\
    &\hphantom { \mathrm {s.t.}~}N\leq N_G, N\in \mathbb {N}_{+} \label{approxTc}
	\end{align}
	\label{approxT}
\end{subequations}
The non-trivial constraint (\ref{approxTb}) on $p_{out}$ is a complicated function about $N$. With the help of the following two lemmas, we will get the range of $N$ equivalent to  (\ref{approxTb}).
\begin{table}[]
\centering
\label{table2}
\caption{$N_{\Omega}$  equivalent to (\ref{g}) under different conditions}
\resizebox{0.5\textwidth}{!}{
\begin{tabular}{cccccc|c}
\hline
$c$     & $d$     & $g^{(1)}(t_0)$    & $g(t_{g})$           & $g(t_{l})$           & $e$                   & $\text{The range of}\ N_{\Omega}$   \\
\hline
$\geq0$ & $\geq0$ & \multirow{2}{*}{} &                       & \multirow{4}{*}{}     & $\geq0$               & \multirow{5}{*}{$(0,\infty)$}\\
$\geq0$ & $<0$    &                   & $\geq0$               &                       &                       &                        \\
$<0$    &         & $\geq0$           &                       &                       & $\geq0$               &                        \\
$<0$    & $\leq0$ &              & $\geq0$               &                       &                       &                        \\
$<0$    & $>0$    & $<0$              & $\geq0$               &                       & $\geq0$               &                        \\
\hline
$\geq0$ & $\geq0$ & \multirow{2}{*}{} &                       & \multirow{4}{*}{}     & $<0$                  & \multirow{6}{*}{$\left(0, \frac{1}{n_m^2}\right]$} \\
$\geq0$ & $<0$    &                   &                  &                       & $\leq0$               &                        \\
$<0$    &         & $\geq0$           &                       &                       & $<0$                  &                        \\
$<0$    & $\leq0$ &              &                 &                       & $\leq0$               &                        \\
$<0$    & $>0$    & $<0$              &                       & $\leq0$               &                       &                        \\
$<0$    & $>0$    & $<0$              & $\geq0$               &                  & $<0$                  &                        \\
\hline
$\geq0$ & $<0$    &                   & $<0$                  &                       & $>0$                  & \multirow{3}{*}{$\left(0, \frac{1}{n_m^2}\right]\cup\left[\frac{1}{n_{m-1}^2}, \infty\right)$} \\
$<0$    & $\leq0$ &             & $<0$                  &                       & $>0$                  &                        \\
$<0$    & $>0$    & $<0$              & $<0$                  &                  & $\geq0$               &                        \\
\hline
$<0$    & $>0$    & $<0$              & $<0$                  & $>0$                  & $<0$                  & $\left(0,\frac{1}{n_m^2}\right]\cup\left[\frac{1}{n_{m-1}^2}, \frac{1}{n_{m-2}^2}\right]$
\\
\hline
\end{tabular}}
\vspace{-10mm}
\end{table}
\begin{lemma}
\label{lemma1}
Let $n=\frac{1}{\sqrt{N}}$, the constraint (\ref{approxTb}) is equivalent to the following two inequations
%\begin{align}
%g(n)\triangleq an^4+bn^3+cn^2+dn+e\geq0 \text{ with }\ 0<n\leq\frac{\sqrt{\mu_1^2+4B\mu_0}-\mu_1}{2B}\ \text{ and } \mu_0>0, \label{g}
%\end{align}
\begin{subequations}
\begin{align}
&g(n)\triangleq an^4+bn^3+cn^2+dn+e\geq0, \label{ga}\\
&0<n\leq\frac{\sqrt{\mu_1^2+4B\mu_0}-\mu_1}{2B}\ \text{with } \mu_0>0, \label{gb}
\end{align}
\label{g}
\end{subequations}
where $a\triangleq B^2>0$, $b\triangleq 2B\mu_1>0$, $c\triangleq \mu_1^2-2B\mu_0 -\left(\Phi^{-1}(\zeta)\right)^2\mu_4$, $d\triangleq -2\mu_0\mu_1-\left(\Phi^{-1}(\zeta)\right)^2\mu_3$, and $e\triangleq \mu_0^2-\left(\Phi^{-1}(\zeta)\right)^2\mu_2$.
Furthermore, let $g^{(1)}(x)$ and $g^{(2)}(x)$ denote the first and second derivatives of $g(x)$, and $t_0$ is the unique positive root of $g^{(2)}(x)=0$, then the changing trend of $g(n)$ with $n>0$ under different conditions have three cases:

Case A: if $c,d\geq0$ or $c<0$, $g^{(1)}(t_0)\geq0$, $g(n)$ monotonically increases with $n>0$.

Case B: if $c\geq0$, $d<0$ or $c<0, d\leq0$, $g(n)$ decreases first and then increases with $n$.

Case C: if $c,g^{(1)}(t_0)<0, d>0$, $g(n)$ increases first, decreases and then increases with $n$.

\end{lemma}
\begin{proof}
The proof is provided in Appendix \ref{appendixl0}.
\end{proof}

\begin{lemma}
\label{lemma2}
Inequality (\ref{ga}) is equivalent to a constraint on the range of $N$ satisfying $N_{\Omega}$ represented in Table II.
%\begin{equation}
%\left\{
%             \begin{array}{lr}
%             \left(0, \frac{1}{n_m^2}\right]& \text{$c,d\geq0$, $e<0$} \\ & \text{or  }\text{$c,e<0$, $g^{(1)}(t_0)\geq0$}\\
%             \left(0, \frac{1}{n_m^2}\right]\cup\left[\frac{1}{n_{m-1}^2}, \infty\right)& \text{$c\geq0$, $d<0$, $g(t_1)\leq0$}\\ & \text{or  }\text{$c,g^{(1)}(t_0)<0,g(t_1)\leq0$}\\
%             \varnothing & \text{otherwise}
%             \label{L0}
%             \end{array}
%\right.
%\end{equation}
where $n_i, i=1,2,\cdots,m$ stands for all the positive and distinguish real roots of $g(x)=0$, which are put in an increasing order. $t_g$ and $t_l$ are the largest and second largest real roots of $g^{(1)}(x)=0$. Especially, $t_g$ can be represented as
\begin{equation}
t_g=\left\{
             \begin{array}{lr}
             t_{11}, & \Delta>0 \\
%             \max\left\{-\frac{b}{3a}+2\sqrt[3]{-\frac{v}{2}}, -\frac{b}{3a}-\sqrt[3]{-\frac{v}{2}}\right\}, & \Delta=0\\
             \max\left\{t_{11}, t_{12} ,t_{13}\right\}, & \Delta\leq0
             \end{array}
\right.
\label{t0}
\end{equation}
\begin{equation}
\text{with }\left\{
\begin{array}{lr
}
t_{11}=-\frac{b}{3a}+\sqrt[3]{-\frac{v}{2}+\sqrt{\Delta}}+\sqrt[3]{-\frac{v}{2}-\sqrt{\Delta}}\\ t_{12}=-\frac{b}{3a}+w\sqrt[3]{-\frac{v}{2}+\sqrt{\Delta}}+w^2\sqrt[3]{-\frac{v}{2}-\sqrt{\Delta}}\\ t_{13}=-\frac{b}{3a}+w^2\sqrt[3]{-\frac{v}{2}+\sqrt{\Delta}}+w\sqrt[3]{-\frac{v}{2}-\sqrt{\Delta}}
\end{array}
\right.
\label{t0i}
\end{equation}
where $\Delta=\frac{v^2}{4}+\frac{u^3}{27}$, $w\triangleq\frac{-1+\sqrt{3}j}{2}$, $u\triangleq\frac{8ac-3b^2}{16a^2}$, and  $v\triangleq\frac{8a^2d-4abc+b^3}{32a^3}$.
\end{lemma}
\begin{proof}
The proof is provided in Appendix \ref{appendixl1}.
\end{proof}
%As $\zeta\leq0.5$, we have $\frac{B}{N}-m_{\bar R_s}\le 0$. Let $n=\frac{1}{\sqrt{N}}$, and then
%\begin{align}
%    f(n)=Bn^2+\mu_1n-\mu_0\leq0,
%    \label{fn}
%\end{align}
%where $\mu_0\triangleq p_0-q_0>0$, $\mu_1=\frac{p_1Q^{-1}(\bar \epsilon)+q_1Q^{-1}(\bar \delta)}{\ln 2}>0$. The equation $f(n)=0$ has two roots $n_{1,2}=\frac{\pm\sqrt{\mu_1^2+4B\mu_0}-\mu_1}{2B}$. As $N>0$, the range of $N$ satisfying (\ref{fn})  can be expressed as
%\begin{align}
%N\geq N_L\triangleq\frac{4B^2}{\left(\sqrt{\mu_1^2+4B\mu_0}-\mu_1\right)^2}
%\end{align}
%Note that the range of $N$ satisfying (\ref{approxTb}), denoted by (\ref{L0}), could be an empty set in some cases, which indicates that the problem is infeasible. The intuition for this case can be interpreted as follows. From (\ref{equation3}), for arbitrary $R_0>0$, $\mathrm{Pr}(R_s<R_0)\geq\zeta_0$ with $0<\zeta_0<1 $. When the preset $\zeta<\zeta_0$, $P_{out}<\zeta$ cannot be satisfied, the problem becomes infeasible.

With the above two lemmas, we now have a more explicit constraint on the range of $N$.
Together with  (\ref{approxTc}), (\ref{approxT}) becomes a one-dimensional optimization problem with explicit constraints. The objective function cannot be further deduced and analyzed for the reason that the parameters in $ m_{\bar R_s}$ and $\sigma_{\bar R_s}$ need to be calculated by numerical integration. We have to conduct a one-dimensional search to obtain the optimal blocklength that maximizes $T$.

\begin{theorem}
\label{theorem0}
The optimal blocklength $N^\#$ for problem (\ref{approxT}) can be obtained by one-dimensional positive integer search on the range of $N\in N_{\omega}\cup N_{\Omega}$, with $N_{\omega}=\left[N_L,N_G\right]$ if $\mu_0>0$, where $N_L\triangleq\frac{4B^2}{\left(\sqrt{\mu_1^2+4B\mu_0}-\mu_1\right)^2}$.
\end{theorem}
\begin{proof}
The proof is finished by replacing (\ref{approxTb}) with conclusions in Lemma \ref{lemma1} and \ref{lemma2}.
\end{proof}

Considering that the blocklength for short-packet transmission usually involves several hundred bits, it will not bring in  too much computational burden. Although the above proposed approach based on distribution approximation cannot provide closed-form expressions for the optimal blocklength $N^\#$ and the optimal effective throughput $T^\#$, it provides a general analytical framework for determining $N^\#$ and analyzing the system performance under any system conditions.

\section{Analysis and Optimization Framework in the High-SNR Domain}
\label{IV}
In Section \ref{III}, an general solution for the optimal blocklength and effective throughput of the system is given. However,
%this method has several limitations. The parameter solving process within this method requires numerical integration, and
we cannot get closed-form expressions for the optimal blocklength and effective throughput, which makes it challenging to get more insights into the impact of system parameters on blocklength selection and system performance.
%In addition, the above method has a slightly larger approximation error in the high-SNR domain, which is not conducive to estimating the optimal blocklength with higher accuracy to maximize the effective throughput of the system.

Fortunately, in the high-SNR domain, the form of $\bar R_s$ in (\ref{equation2}) will be greatly simplified, making it possible to obtain the analytical expression of $p_{out}$.
In this section, we will deduce the CDF $F_{\bar R_s}(r)$ with the assumption of high SNR. Further, closed-form expressions for outage probability $p_{out}$ and the effective throughput $T$ are derived. It can be proved that $T$ is a quasi-concave function of $N$, so that the optimal $N$ which can maximize $T$ will be found by utilizing a bisection search.

\subsection{High-SNR Approximation}
As discussed above, the derivation of the distribution of $\bar R_s$ is challenging since the form of $\bar R_s$ is complicated. Fortunately, the challenge will be tackled in high-SNR domain. It is reasonable to adopt the high-SNR assumption for the reason that complicated signal processing cannot be performed on resource-limited IoT nodes, so sufficient SNR should be guaranteed to ensure the quality of the received signal so that the information can be decoded correctly at these nodes.
% Meanwhile, the communication environment for short distance transmission is relatively simple, which means that the average SNRs at the receiving ends will not be too low.
\begin{figure}[!t]
\begin{center}
\includegraphics[width=3 in]{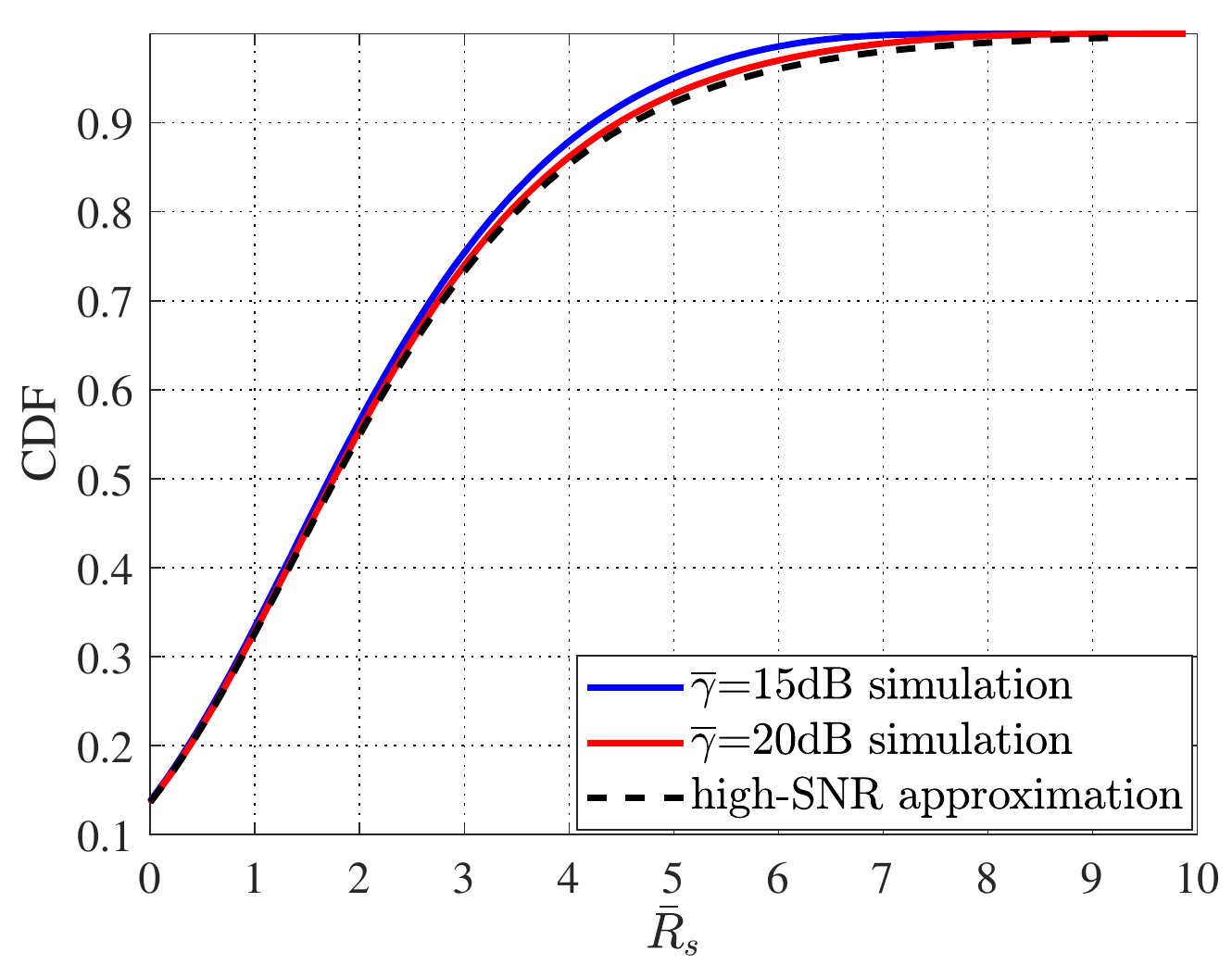}
\end{center}
\vspace{-8mm}
\caption{\small{CDFs of $\bar R_s$ and their high SNR approximations under different $\overline\gamma$ with $N=200$, $\bar\epsilon=\bar\delta=10^{-3}$, $K=4$.}}
\label{fig3}
\vspace{-10mm}
\end{figure}
%This section will derive the $f_{R_s}\left(r\right)$ for the MISO system
In the high-SNR domain, $\log_2(1+x)\approx \log_2 x$, whereas $V_i=1-\left(1+\gamma_i\right)^{-2}\approx 1, i\in{b, e}$,
then the achievable secrecy rate in (\ref{equation2}) for the high SNR regime can be respectively approximated as
\begin{align}
\tilde R_s=\log_2\gamma_b-\log_2\gamma_e-P-Q=\tilde C_s-\upsilon, \label{equationR}
\end{align}
%Then we have
%\begin{align}
%\bar R_s\approx \tilde R_s\triangleq C_s-\upsilon,
%\label{equationR}
%\end{align}
where $\tilde C_s \triangleq \log_2\gamma_b-\log_2\gamma_e$, $\upsilon\triangleq \frac{t}{\sqrt{N}}$ with $t\triangleq\frac{Q^{-1}\left(\bar \epsilon\right)+Q^{-1}\left(\bar \delta\right)}{\ln2}\geq0$ (as $\bar\epsilon$ and $\bar\delta$ are no more than 0.5). Obviously, in the high-SNR domain, if blocklength $N$, reliable and secure constraints $\bar \epsilon$ and $\bar \delta$ are given, $\bar R_s$ is equivalent to adding a constant term to  $\tilde C_s$. Recalling the PDFs of $\gamma_b$ and $\gamma_e$, the CDF of $\tilde R_s$ in the high SNR domain will be deduced as Lemma \ref{lemma3} shows.
\begin{lemma}
\label{lemma3}
For the high-SNR regime, the CDF of $\tilde R_s$ can be represented as
\begin{align}
F_{\tilde R_s}\left(r\right)=\left(\frac{\mathrm{2}^{r+\upsilon}}{\mathrm{2}^{r+\upsilon}+1}\right)^K.
\label{deducedpout}
\end{align}
\label{lemma}
\end{lemma}
\vspace{-10mm}
\begin{proof}
The proof is provided in Appendix \ref{appendix0}.
\end{proof}

%Now we have obtained a closed-form approximation of $F_{\bar R_s}(r)$.
Fig. \ref{fig3} depicts the accurate $F_{\bar R_s}(r)$ under different %average SNRs
$\overline\gamma$, as well as the high-SNR approximation of $F_{\tilde R_s}(r)$. It can be observed that the approximate CDF curve is below the accurate $F_{\bar R_s}(r)$. Meanwhile, the approximate accuracy gradually improves as $\overline\gamma$ increases. When $\overline\gamma=20\text{dB}$, the approximate curve is almost overlap with the accurate one.

\subsection{Unconstrained Effective Throughput Optimization}
\label{Unconstrained Effective Throughput}
Based on (\ref{deducedpout}), the outage probability and the effective throughput can be represented as
%Substituting (\ref{deducedpout}) into the outage probability defined in (\ref{equation3}) yields the outage probability as
\begin{align}
p_{out}&=F_{\tilde R_s}\left(\frac{B}{N}\right)=\left(\frac{\mathrm{2}^{\frac{B}{N}+\frac{t}{\sqrt{N}}}}{\mathrm{2}^{\frac{B}{N}+\frac{t}{\sqrt{N}}}+1}\right)^K=H(N)^K,
\label{pout}\\
T&=\frac{B}{N}\left[1-H(N)^K\right] \label{equation20}
\end{align}
\vspace{-2mm}
%the corresponding effective throughput can be represented as
%\begin{align}
%%T=\frac{B}{N}\left[1-\left(\frac{\mathrm{2}^{\frac{B}{N}+\frac{t}{\sqrt{N}}}}{\mathrm{2}^{\frac{B}{N}+\frac{t}{\sqrt{N}}}+1}\right)^K\right].
%T=\frac{B}{N}\left[1-H(N)^K\right]
%\label{equation20}
%\end{align}
where $h(N)\triangleq \frac{B}{N}+\frac{t}{\sqrt{N}}$ and $H(N)\triangleq\frac{2^{h(N)}}{2^{h(N)}+1}$. Relax the integer $N$ as a positive real number,
%In order to facilitate the analysis of the effect of variable $N$ on $T$, we first relax the integer  blocklength $N$ as a positive real number.
then the first derivative of $p_{out}$ to $N$ is
\begin{align}
\frac{dp_{out}}{dN}=-\frac{\lambda(N)}{N^2\left(1+2^{h(N)}\right)}H(N)^K\leq0, \nonumber
\end{align}
with $\lambda(N)\triangleq\ln2 \left(B+\frac{t}{2}\sqrt{N}\right)K$, which implies that  $p_{out}$ is a monotonically decreasing function about $N$. In addition, it can be proved easily that $p_{out}$  monotonically increases with
$B$ while decreases with $\bar \epsilon$ and $\bar \delta$.
Then  the following theorem proves that $T$ is a quasi-concave function about the relaxed $N$.

\begin{theorem}
The effective throughput $T$ in (\ref{equation20}) is a quasi-concave function of the relaxed continuous $N$. The optimal blocklength  that yields the largest effective throughput without the constraint $p_{out}\leq\zeta$ can be chosen in $\{\lceil {N^*} \rceil ,\lfloor {N^*} \rfloor \}$, where ${N^*}>0$ is the unique root of
\begin{align}
\Xi\left(N\right)\triangleq\left(\frac{\lambda(N)}{N\left(2^{h(N)}+1\right)}+1\right)\left(H(N)\right)^K-1=0.
\label{equation24}
\end{align}
\label{theorem2}
\end{theorem}
\vspace{-7mm}
\begin{proof}
The proof is provided in Appendix \ref{prooftheorem1}.
\end{proof}

According to Theorem \ref{theorem2}, $N^*$ that satisfies (\ref{equation24}) can be searched by using the bisection method within $N\geq1$. % It is a computational efficient approach to obtain the optimal solution.
Then the following corollary provides some insights into the behavior of $N$.
\newtheorem{corollary}{Corollary}
\begin{corollary}
The optimal $N^*$ for (\ref{equation13a}) increases with $B$, while it decreases as $K$ increases.
\label{corollary1}
\end{corollary}
\begin{proof}
The proof is shown in Appendix \ref{appendix4}.
\end{proof}
%After obtaining $N^*$according to Theorem \ref{theorem2}
Then the optimal value for (\ref{equation13a}) denoted by $T(N^*)$ can be evaluated by substituting $N^*$ in (\ref{equation20}). The following corollary will provide the impact of the related system parameters on $T(N^*)$.
\begin{corollary}
The optimal $T(N^*)$ of (\ref{equation13a}) increases with $B$, $\bar\epsilon$, $\bar\delta$ and $K$.
\label{corollary2}
\end{corollary}
\begin{proof}
The proof can be seen in Appendix \ref{appendix5}.
\end{proof}
In addition, it can be observed from (\ref{equation20}) that $T$ obtained by high-SNR approximation is independent with
%the average SNR
$\overline\gamma$, so neither $N^*$ nor $T(N^*)$ is influenced by $\overline\gamma$. It means that rising $\bar \gamma$ cannot improve the performance in the high SNR domain.

\subsection{Optimization under the Outage Constraint}
The above optimal effective throughput is derived without the outage constraint. To go a step further,
%we should take the outage constraint (\ref{equation13b}) into consideration.
the following corollary shows the optimal blocklength with the outage constraint.% for the high-SNR regime.

\begin{corollary}
In the high-SNR domain, the optimal blocklength with the outage constraint is
%for problem (\ref{equation22}) is given by
\begin{align}
N^\#= \begin{cases}
\mathop{\arg\max}\limits_{N\in \left \{{\lceil N^* \rceil ,\lfloor N^* \rfloor }\right \}}\,T(N), &\text{$N_0 \leq N^* \leq N_G$,$\sqrt[K]{\zeta}\geq\frac{1}{2}$ }\\
\left \lceil{N_0}\right \rceil ,&\text{$N^* \leq N_0 \leq N_G $,$\sqrt[K]{\zeta}\geq\frac{1}{2}$ }\\
N_G, &\text{$N_0 \leq N_G \leq N^* $,$\sqrt[K]{\zeta}\geq\frac{1}{2}$ }\\
\text{No Value} ,& \text{Otherwise}
\end{cases}\label{equation27}
\end{align}
where $\Delta_2\triangleq t^2+4B\log_2\left(\frac{\sqrt[K]{\zeta}}{1-\sqrt[K]{\zeta}}\right)$, and $N_0\triangleq \frac{4B^2}{(\sqrt{\Delta_2}-t)^2}$.
\label{corollary3}
\end{corollary}
\begin{proof}
By substituting (\ref{pout}) into (\ref{equation13b}), we have
\begin{align}
\frac{B}{N}+\frac{t}{\sqrt{N}}-\log_2\left(\frac{\sqrt[K]{\zeta}}{1-\sqrt[K]{\zeta}}\right) \leq0.
\label{c31}
\end{align}
Accordingly, the range of $N$ satisfying ($\ref{c31}$) is
%\begin{align}
%N\in\left[N_0, \infty\right), \text{with}\qquad \sqrt[K]{\zeta}\geq\frac{1}{2}.
%\label{c32}
%\end{align}
\begin{equation}
\left\{
             \begin{array}{lr}
             N\in\left[N_0, \infty\right), & \text{$\sqrt[K]{\zeta}\geq\frac{1}{2}$ } \\
             N\in\varnothing, & \text{Otherwises}
             \label{c32}
             \end{array}
\right.
\end{equation}
The proof is completed by substituting (\ref{equation13c}) and (\ref{c32}) into Theorem \ref{theorem2}.
\end{proof}

Then the optimal effective throughput
$T^\#=T(N^\#)$ can be calculated by substituting $N^\#$ into (\ref{equation20}). Notably, %the optimal blocklength
$N^\#$ is unavailable when $\zeta<2^{-K}$ or $N_0>N_G$, which is due to the excessively strict outage constraint. Specifically, $\zeta<2^{-K}$ can be interpreted as the situation that current system  cannot meet the outage restriction $p_{out}<\zeta$ under the preset reliable and secure constraints $\bar\epsilon$ and $\bar\delta$, while $N_0>N_G$ refers to the situation where the blocklength $N_0$ is increased to meet the outage restriction but exceeds the blocklength range of short packets.

\section{Numerical Simulation}
\label{NumericalSimulation}
In this section, we provide numerical results to show the reliable and secure performance of short-packet communication systems. The parameter settings are as follows, unless otherwise specified: $\overline\gamma = 10 \text{dB}$ ($\overline\gamma = 20 \text{dB}$ in the high SNR domain), $\bar\epsilon=10^{-3}$, $\bar\delta= 10^{-3}$, $K = 8$, $B =400$, $ N_G=1000$. All the simulation results shown in this paper are obtained by averaging over $1, 000, 000$ channel realizations.
\subsection{Distribution Approximation}

Fig. \ref{fig4} plots the effective throughput $T$ versus the blocklength $N$ and the transmitted bit number per block $B$ without restriction of outage probability, where camber (a) shows the simulation result, while camber (b) presents the approximate deviation of $T$ obtained by distribution approximation method.
For each $B$, the optimal $N$ and $T$ obtained by approximation is given in points (c). From Fig. \ref{fig4}, it is confirmed that the optimal $T$ can be obtained by carefully selecting $N$ with a given $B$. Moreover, as can be observed in camber (b), the approximate error of $T$ is within an acceptable range intuitively. Consequently, the optimal blocklength $N$ and the corresponding $T$ obtained by one-dimensional search with the approximation result as shown in (c) fit well with the simulation result.
%
%Fig.\ref{fig4} confirms the conclusion that the secrecy throughput is a quasi-concave function of the relaxed continuous blocklength given in Theorem \ref{theorem2}, and the method proposed in Theorem \ref{theorem2} can be used to accurately find the optimal blocklength for secrecy throughput maximization. Additionally, the optimal blocklength is increase with $B$ while the optimal secrecy throughput is also an increasing function of $B$.
\begin{figure}[!t]
\begin{center}
\includegraphics[width=3 in]{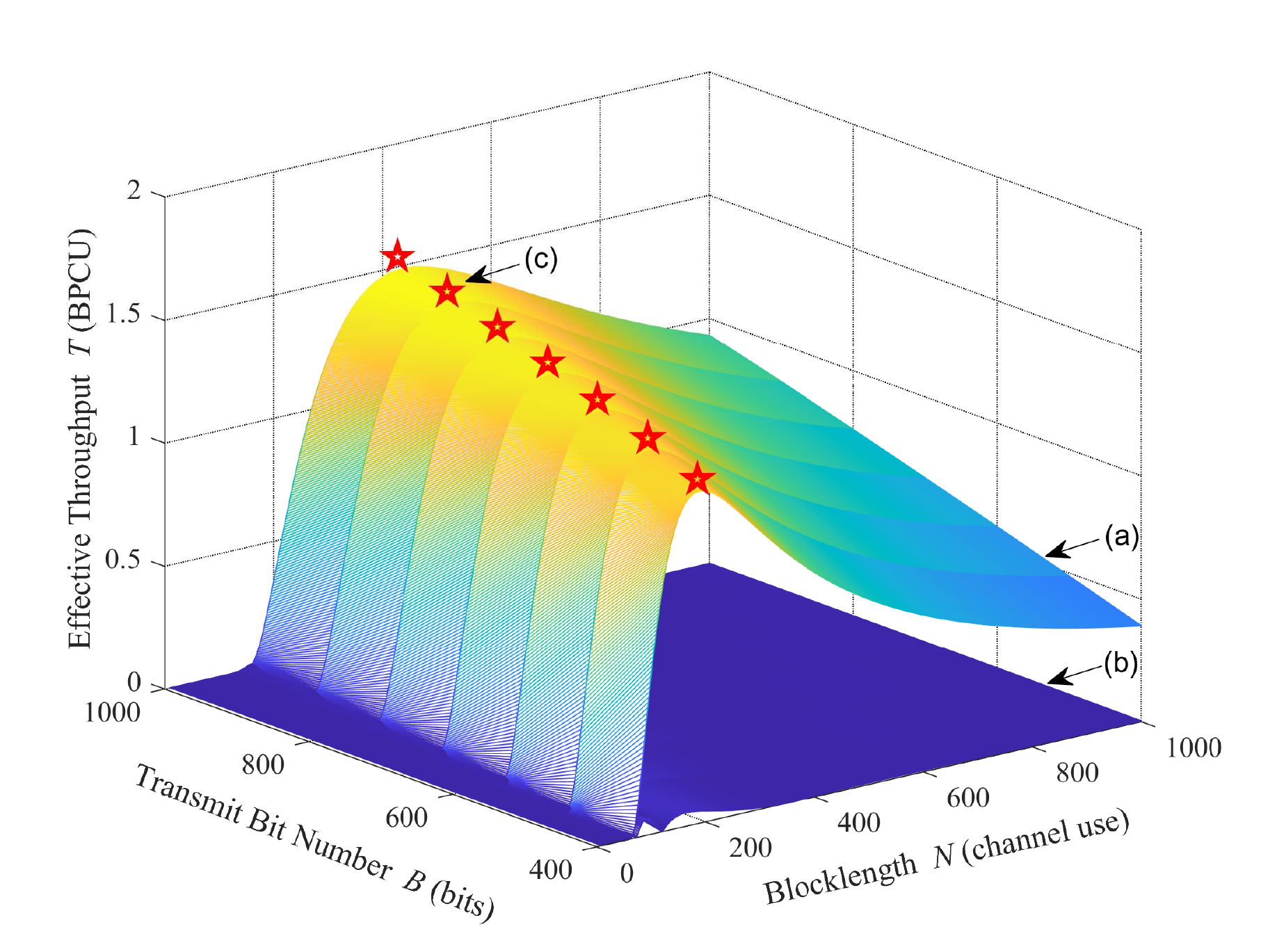}
\end{center}
\vspace{-7mm}
\caption{\small{The effective throughput $T$ without the outage constraint versus $N$ and $B$ with $\overline\gamma=10\text{dB}$. Camber (a) is the simulation results of $T$. Camber (b) is the deviation of $T$ between simulation and approximation result. Points (c) are the optimal blocklength obtained by One-dimensional search and effective throughput under a certain $B$.}}
\label{fig4}
\vspace{-5mm}
\end{figure}

\begin{figure}[!t]
\begin{center}
\includegraphics[width=3 in]{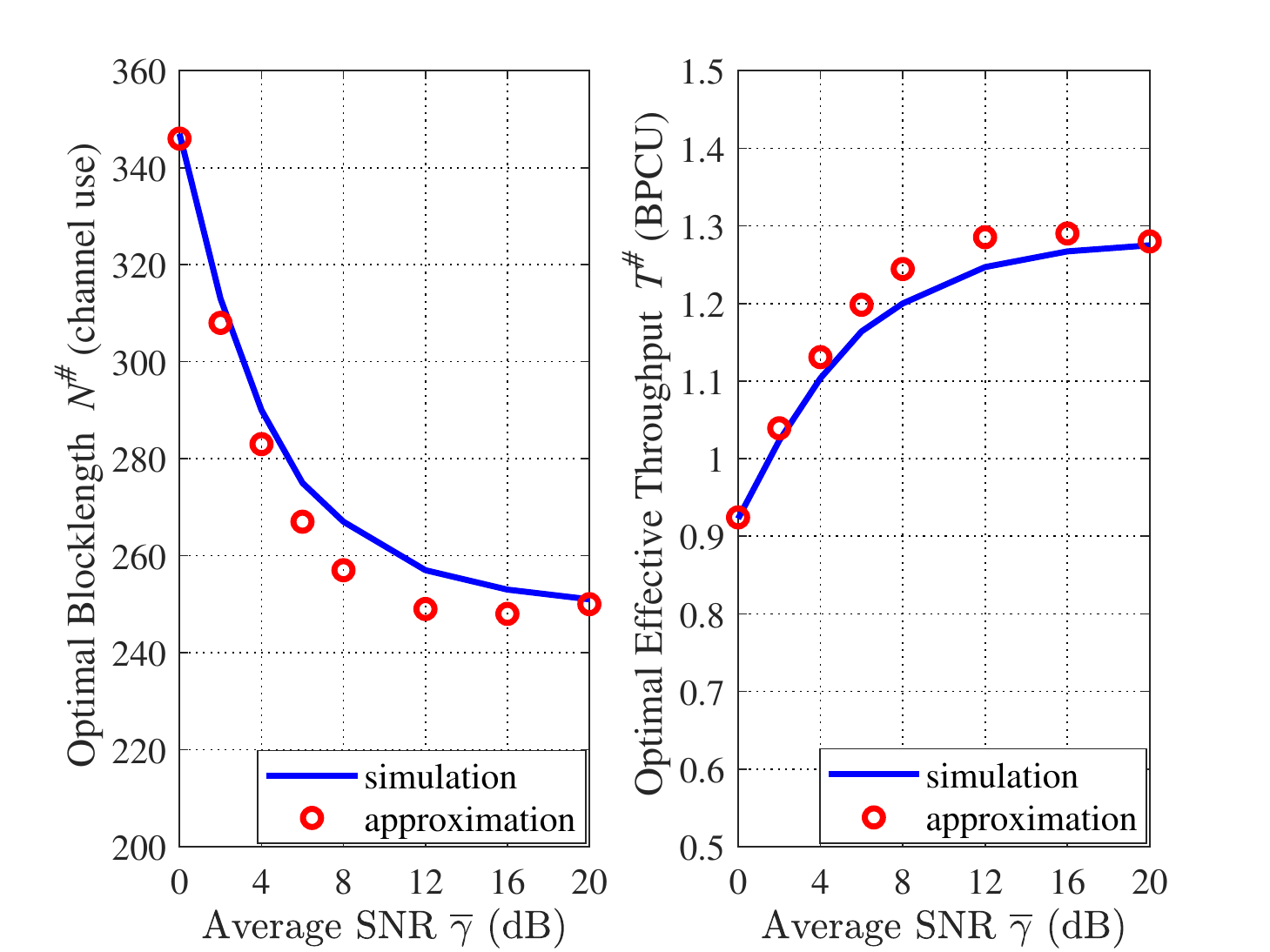}
\end{center}
\vspace{-5mm}
\caption{\small{The optimal blocklength $N^\#$ and effective throughput $T^\#$ versus the average SNR $\overline\gamma$ ($\zeta=0.2$).}}
\label{fig5}
\vspace{-10mm}
\end{figure}
Fig. \ref{fig5} shows the optimal blocklength $N^\#$ and the corresponding optimal effective throughput $T^\#$ versus the average SNR $\overline\gamma$ with the outage-probability constraint. As shown in the figure, there is a gap between simulation and approximation results, which comes from the approximation error introduced by using the distribution approximation technique to approximate the CDF of $R_s$. Specifically, letting $N_{simu}$ and $N_{appro}$ denote the simulation and approximation results, respectively, the maximum relative error of $N^\#$ defined as $\frac{|N_{simu}-N_{appro}|}{N_{simu}}\times 100\%$ is $3.75\%$, while the maximum relative deviation of $T^\#$ is $3.70\%$, which are acceptable. That is, the results obtained by distribution approximation method roughly match the simulation results. In addition, we can observe from Fig. \ref{fig5} that the higher $\overline\gamma$ the lower $N^\#$, as well as the higher $T^\#$. The essential reason for this phenomenon is that as the signal quality improves, the coding rate can be appropriately increased. That is, reduce the blocklength to increase the effective throughput. What's more, as $\overline\gamma$ goes up, the change trend with $N^\#$ and $T^\#$ slows down, which is consistent with the conclusion that neither $N^\#$ nor $T^\#$ is influenced by $\bar\gamma$ in the high SNR domain.
\subsection{High SNR Domain}
\begin{figure}[!t]
\begin{center}
\includegraphics[width=3 in]{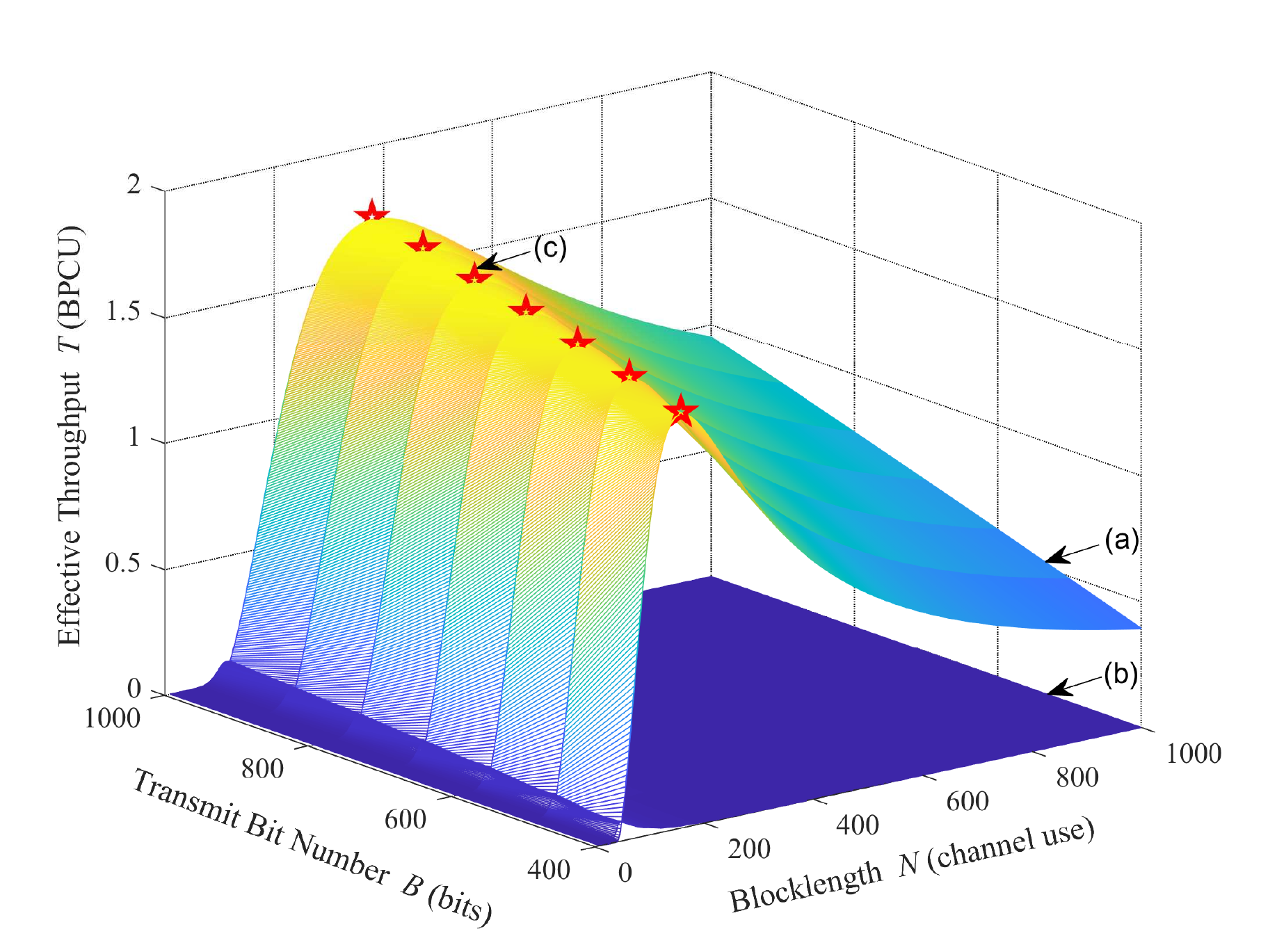}
\end{center}
\vspace{-5mm}
\caption{\small{The effective throughput $T$ without the outage constraint versus $N$ and $B$ with $\overline\gamma=20\text{dB}$. Camber (a) is the simulation results of $T$. Camber (b) is the deviation of $T$ between simulation and approximation result. Points (c) are the optimal blocklength obtained by Binary search and effective throughput under a certain $B$.}}
\label{fig6}
\vspace{-5mm}
\end{figure}
Fig. \ref{fig6} depicts the approximate effect of effective throughput $T$ in the high SNR domain without the outage constraint. The deviation of $T$ between simulation and approximation results is shown in camber (b) and it is obvious that the deviation is negligible. Fig. \ref{fig6} confirms the conclusion that $T$ is a quasi-concave function of the relaxed continuous blocklength given in Theorem \ref{theorem2}, and the method proposed in Theorem \ref{theorem2} can be used to accurately find the optimal blocklength for effective throughput maximization, as ($c$) shows, the optimal blocklength $N^*$ obtained by binary search coincide well with the simulation result.

\begin{figure}[!t]
\begin{center}
\includegraphics[width=3 in]{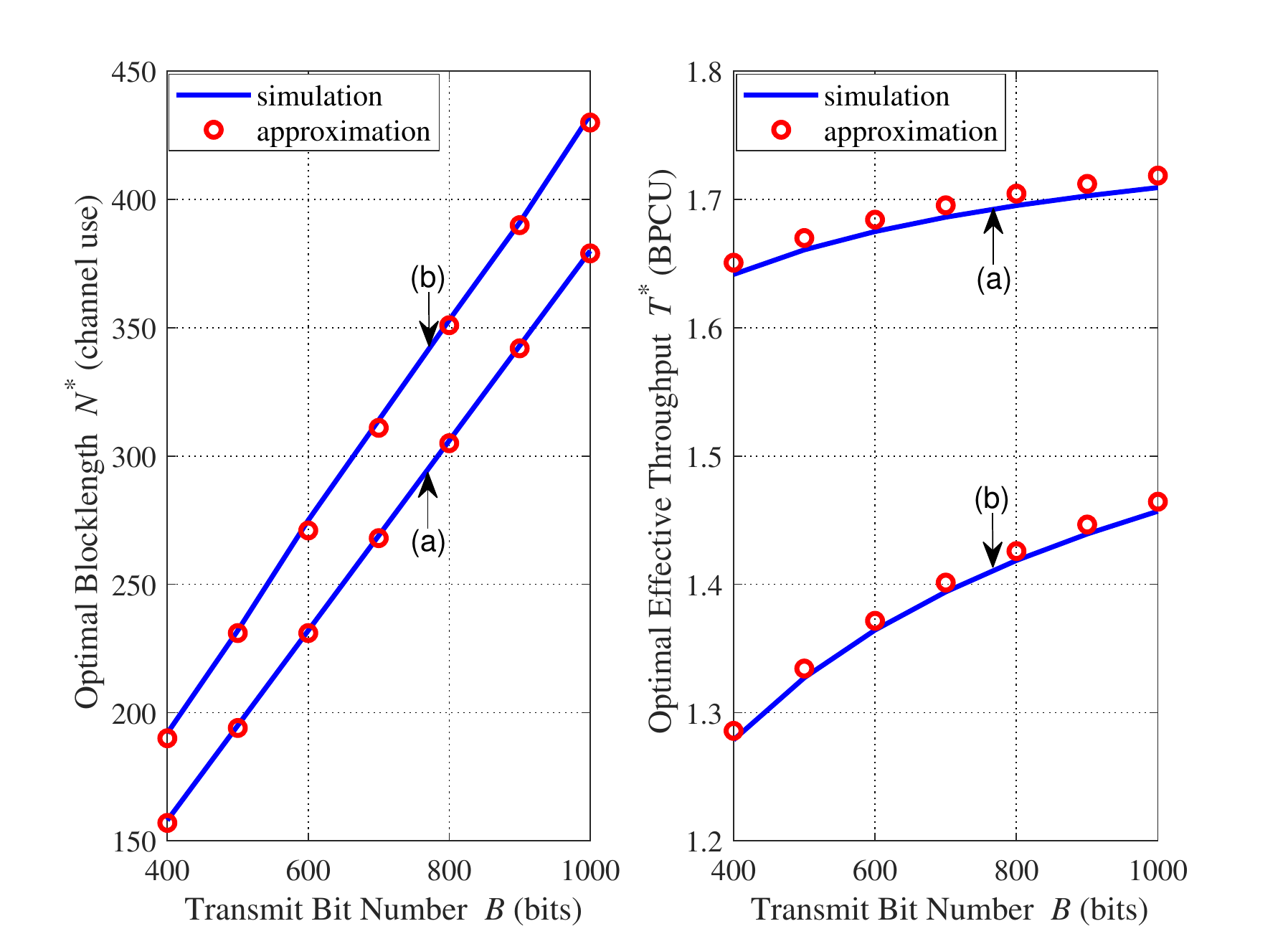}
\end{center}
\vspace{-5mm}
\caption{\small{The optimal blocklength $N^*$ and the effective throughput $T^*$ without the outage constraint versus transmitted bit number per block $B$, the reliable constraint $\bar\epsilon$ and the secure constraint $\bar\delta$. (a) $\bar\epsilon=10^{-1},\bar\delta=10^{-1}$, (b) $\bar\epsilon=10^{-5},\bar\delta=10^{-5}$.}}
\label{fig7}
\vspace{-10mm}
\end{figure}
\begin{figure}[!t]
\begin{center}
\includegraphics[width=3 in]{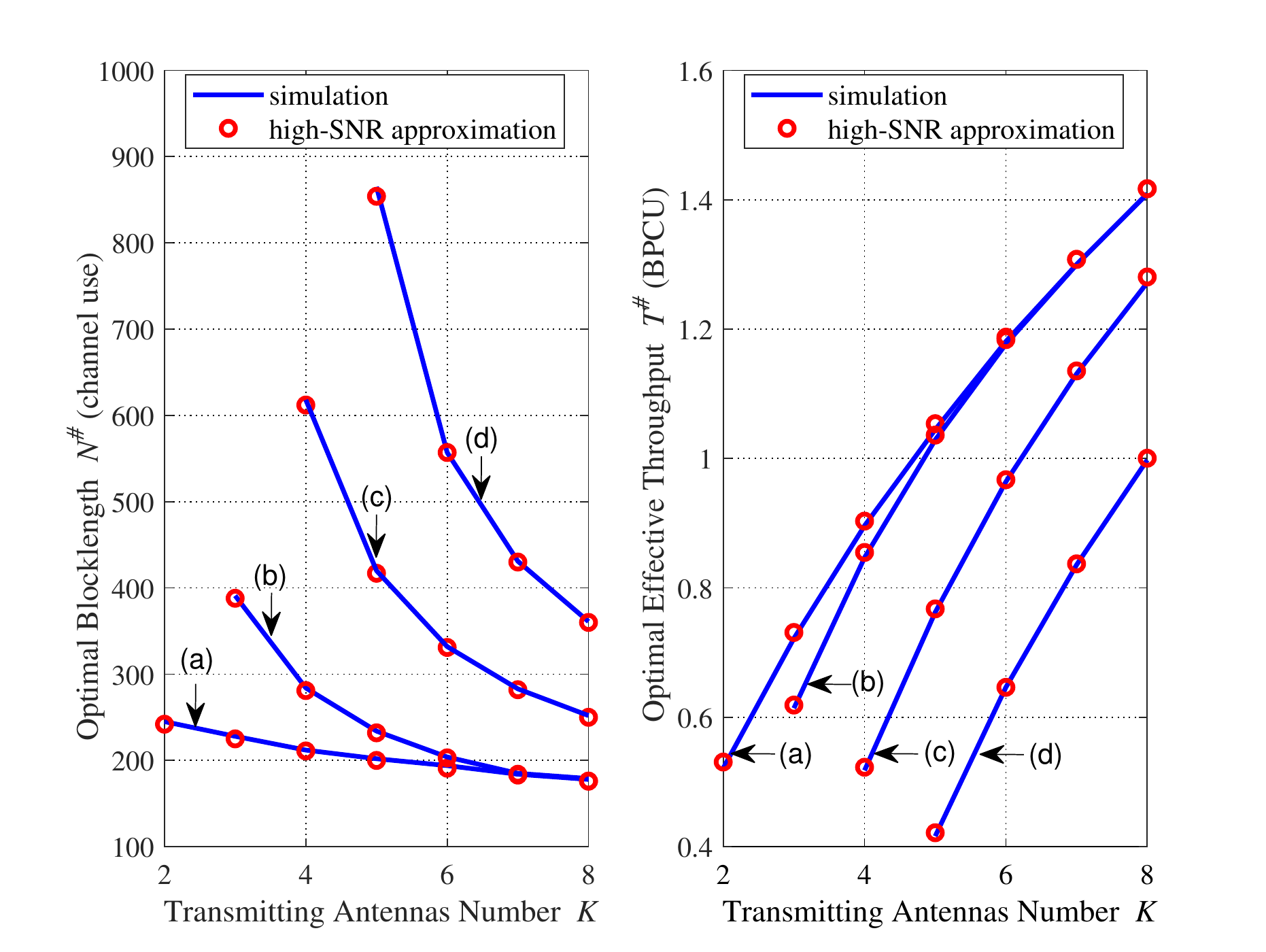}
\end{center}
\vspace{-5mm}
\caption{\small{The optimal blocklength $N^{\#}$ and the optimal effective throughput $T^\#$ %obtained from corollary 3
versus the number of transmitting antennas $K$. (a) Without restriction of outage probability, (b) $\zeta=0.4$, (c) $\zeta=0.2$, (d) $\zeta=0.1$.}}
\label{fig8}
\vspace{-10mm}
\end{figure}
In Fig. \ref{fig7}, without the outage constraint, the impacts of transmitted bit number per block $B$, as well as the reliable and secure constraints $\{\bar\epsilon, \bar\delta\}$ on $N^*$ and $T^*$ are presented, respectively. As can be observed, the simulation results coincide well with the analytical results in the high SNR domain. It is shown in Fig. \ref{fig7} that an increasing $B$ requires more channel uses, which is consistent with intuition and confirms the relative theoretical analysis in Corollary \ref{corollary1}. Furthermore, comparing the two cases of (a) and (b), it can be observed from Fig. \ref{fig7} that $N^*$ increases as $\bar\epsilon$ and $\bar\delta$ decrease, which indicates that when the reliable or secure constraint is relaxed, Alice can appropriately increase the coding rate. These phenomena are consistent with the theoretical analysis in Corollary \ref{corollary1}, and they all reflect the trade-off between reliable-secure performance and latency performance.

Fig. \ref{fig7} also shows that the optimal effective throughput slowly increases when $B$ increases. It is worthy noting that the increasement of effective throughput cannot be expected by continuously increasing $B$. As a larger $B$ result in the increasement of $N$, which has an upper bound and cannot increase arbitrarily for short-packet communications. Additionally, when the restriction on performance is relaxed, the effective throughput will be improved, which is consistent with the conclusion in Corollary \ref{corollary2}.

%\begin{figure}[!t]
%\begin{center}
%\includegraphics[width=3 in]{N_K_HSNR.eps}
%\end{center}
%\vspace{-5mm}
%\caption{\small{The optimal blocklength $N^{\#}$ obtained from corollary 3 versus the number of transmitting antennas $K$. (a) Without restriction of outage probability, (b) $\zeta=0.4$, (c) $\zeta=0.2$, (d) $\zeta=0.1$.}}
%\label{fig8}
%\vspace{-5mm}
%\end{figure}
%
%\begin{figure}[!t]
%\begin{center}
%\includegraphics[width=3 in]{T_K_HSNR.eps}
%\end{center}
%\vspace{-5mm}
%\caption{\small{The optimal effective throughput $T^\#$  obtained from  Corollary 3 versus the number of transmitting antennas $K$. (a) Without restriction of outage probability, (b) $\zeta=0.4$, (c) $\zeta=0.2$, (d) $\zeta=0.1$.}}
%\label{fig9}
%\vspace{-5mm}
%\end{figure}

Fig. \ref{fig8} plots the optimal blocklength $N{^\#}$ obtained from corollary \ref{corollary3} versus the number of transmitting antennas $K$ under different restriction of outage probability $\zeta$. As depicted in the graph, a larger $K$ causes a decrease with $N^\#$. The reason behind this is that as $K$ increases, the quality of the received signal at Bob improves, which can be used to increase the coding rate and improve the delay performance with shorter blocklength. In addition, comparing  the four cases of ($a$), ($b$), ($c$) and ($d$), we will find that the restriction $p_{so}<\zeta$ will affect the choice of $N^\#$. Specifically, when $K$ is not large enough, $N^*$ obtained by solving (\ref{equation24}) cannot meet the restriction of $p_{so}<\zeta$, and the optimal blocklength should be larger at this time with $N^\#=\left \lceil{N_0}\right \rceil$ according to Corollary \ref{corollary3}. Furthermore, as $\zeta$ decline, in other words, the limitation becomes stricter, the optimal blocklength $N^\#$ should be larger, which is consistent with the analysis in Section \ref{IV} and reflects the trade-off between transmission performance and coding rate.

Fig. \ref{fig8} also shows the impact of transmitting antenna number $K$ on the optimal effective throughput $T^\#$. It can be observed from the figure that $T^\#$ steadily increases as $K$ grows, result from the improvement of channel quality can support greater information throughput, which coincide with the analysis in Corollary  \ref{corollary2}. Comparing the four cases, a tighter outage-probability constraint will decrease the throughput. The reason behind this is that the blocklength is expected to be larger so as to meet the outage constraint at this time, which considerably reduces the rate.

%It can be observed from Fig. \ref{fig10} that when $\overline\gamma=10\text{dB}$, the effect of distribution approximation method is better than that using high-SNR method. Specifically, the maximum relative errors of $N^\#$ obtained by the two methods are $1.34\%$ and $3.23\%$, respectively. In fact, when the average SNR is smaller, the advantage of utilizing distribution approximation will be more outstanding. The maximum relative errors of $T^\#$ obtained with the two methods are  $1.34\%$ and $3.33\%$, that is, the effect of using distribution approximation is better than high-SNR method under moderate and low SNR domain. This shows the potential of distribution approximation method for analysis with moderate and low SNR applications.
\section{Conclusion}
\label{Conclusion}
%\begin{figure}[!t]
%\begin{center}
%\includegraphics[width=3.5 in]{compare.eps}
%\end{center}
%\vspace{-5mm}
%\caption{\small{The optimal blocklength $N^\#$ and effective throughput $T^*$ versus transmit bit number $B$ ($\overline\gamma=10\text{dB}$, $K=8$, $\zeta=0.1$). }}\label{fig10}
%\vspace{-5mm}
%\end{figure}
In this paper, we have analyzed and evaluated the reliable and secure performance of short-packet communication systems. Based on the characteristics of short-packet communications, the paper has proposed the concept of outage probability. On this basis, the effective throughput has further been defined and used as the performance metric to evaluate the reliable and secure performance. In particular, a general analytical framework and the case of high SNR have been respectively proposed and investigated. Furthermore, the effects of system parameters as well as the reliable and secure constraints on the optimal blocklength and the optimal effective throughput have been analyzed.  Numerical results have verified the feasibility and accuracy of the proposed approximation method, and confirm the influence of the main system parameters on the selection of the optimal blocklength and the optimal effective throughput. The approximate framework proposed in this paper is of potential interest in the analysis of the performance of other communication scenarios.
\begin{appendices}
\section{ Proof for Lemma \ref{lemma1}}
\label{appendixl0}
As $\Phi(x)$  monotonically increases with $x$, the equivalent form of (\ref{approxTb}) can be expressed as
\begin{align}
\frac{\frac{B}{N}-m_{\bar R_s}}{\sigma_{\bar R_s}}\leq\Phi^{-1}(\zeta).
\label{A0}
\end{align}
For $\zeta<0.5$, we have $\Phi^{-1}(\zeta)<0$, thus (\ref{A0}) holds when
\begin{align}
\frac{B}{N}-m_{\bar R_s}<0, \label{L1a}
\end{align}
\begin{align}
\left(\frac{B}{N}-m_{\bar R_s}\right)^2>\left(\phi^{-1}(\zeta)\right)^2\sigma_{\bar R_s}^2 \label{L1b}
\end{align}
%\begin{subequations}
%\begin{align}
%&\frac{B}{N}-m_{\bar R_s}\leq0, \label{L1a}\\
%&\left(\frac{B}{N}-m_{\bar R_s}\right)^2\geq\left(\phi^{-1}(\zeta)\right)^2\sigma_{\bar R_s}^2 \label{L1b}
%\end{align}
%\label{L1}
%\end{subequations}
hold at the same time. Further, as $\mu_1>0$, it can be easily deduced that (\ref{L1a}) is equivalent to $0<n<\frac{\sqrt{\mu_1^2+4B\mu_0}-\mu_1}{2B}$ if $\mu_0>0$, while (\ref{L1b}) is equivalent to $g(n)\geq0$.
%which can be obtained by substituting (\ref{equation15}) and (\ref{equation16}) into (\ref{L1b}).
Further, we study the properties of $g(n)$. The derivatives of $g(n)$ can be represented as
\begin{align}
g^{(1)}(n)=4an^3+3bn^2+2cn+d,~ g^{(2)}(n)=12an^2+6bn+2c,~ g^{(3)}(n)=24an+6b, \nonumber
\end{align}
where $g^{(3)}(n)>0$ since $a,b>0$, thus $g^{(2)}(n)$ increases with $n$. Then we will discuss the changing trend of $g(n)$ with $n>0$ under different cases by means of the derivatives of $g(n)$.

1) $c\geq0$:
$g^{(1)}(n)$ increases with $n$ as $g^{(2)}(n)>0$. Thus $g(n)$ monotonically increases with $n$ if $g^{(1)}(0)=d\geq0$, otherwise $g(n)$ decreases first and then increases with $n$ when $d<0$.

2) $c<0$:
It can be easily deduced that the unique positive root of $g^{(2)}(n)=0$ is $t_0=\frac{\sqrt{9b^2-24ac}-3b}{12a}>0$. Accordingly,
$g^{(1)}(n)$ decreases first and then increases with $n$, and will reach the minimum value when $n=t_0$. Similarly, $g(n)$ monotonically increases with $n$ if $g^{(1)}(t_0)\geq0$, it will decreases first and then increases with $n$ when  $d\leq0$, while $g(n)$ increases first, decreases and then increases with $n$ when $g^{(1)}(t_0)<0$ and $d>0$.

The proof is finished with the above discussions.
\section{ Proof for Lemma \ref{lemma2}}
\label{appendixl1}
Based on the the changing trend of $g(n)$ with $n>0$ given in Lemma \ref{lemma1}, the value range of $n$ satisfying (\ref{g}) under different cases can be discussed as follows. Notably, the discussion is conducted within $n>0$.

Case A: $c,d\geq0$ or $c<0$, $g^{(1)}(t_0)\geq0$

As $g(n)$ monotonically increases with $n>0$, $g(n)=0$ has $0$ or $1$ zero crossing point (ZCP).

1) if $g(n)>g(0)=e\geq0$, $g(n)>0$, and the corresponding $N_\Omega\in(0,\infty)$.

2) if $e<0$, $g(n)$ is negative first and then positive with an increasing $n$. To determine the ZCPs of $g(n)$, let $g(n)=0$ and solve the quartic equation, and then four roots including multiple roots, complex roots and negative roots will be obtained. Pick out all the positive real roots with different values and put them in an increasing order, denoted by $n_i, i=1,2,\cdots,m$. Let $n_m$ as the maximum real value root, then the corresponding $N_\Omega\in\left(0, \frac{1}{n_m^2}\right]$.

Case B: $c\geq0$, $d<0$ or $c<0, d\leq0$

$g(n)$ decreases first and then increases with $n$, and reach the minimum value when $g^{(1)}(n)=0$. Let $g^{(1)}(n)=0$ and solve the cubic equation, there must exist a maximum positive real root as $a>0$ and $g(n)$ decreases first and then increases with $n>0$, which can be denoted by $t_g$ as (\ref{t0}) shows. $g(n)=0$ has $0-2$ ZCPs with $n>0$ for different $g(0)$ and $g(t_g)$.

1) if $g(t_g)\geq0$, $g(n)>0$, then the corresponding $N_\Omega\in(0,\infty)$.

2) if $e\leq0$, there exists only one ZCP, $N_\Omega\in\left(0, \frac{1}{n_m^2}\right]$.

3) if $g(t_g)<0$ and $e>0$, there exists two ZCPs, $N_\Omega\in\left(0, \frac{1}{n_m^2}\right]\cup\left[\frac{1}{n_{m-1}^2}, \infty\right)$.

Case C: $c,g^{(1)}(t_0)<0, d>0$

Under this case, $g^{(1)}(n)$ has two ZCPs, represented as $t_{l}$ and $t_{g}$ with $t_{l}<t_{g}$, which can be obtained by (\ref{t0i}). Further, $g(n)$ increases first, decreases and then increases with $n>0$. Thus $g(n)=0$ has $0-3$ ZCPs based on different $g(0)$, $g(t_{l})$ and $g(t_{g})$.

1) if $g(0)=e\geq0$, $g(t_{g})\geq0$, then $g(n)>0$ with $n>0$, $N_\Omega\in(0,\infty)$.

2) if $g(t_{l})\leq0$ or $e<0$, $g(t_{g})\geq0$, there exists one ZCP, then $N_\Omega\in\left(0, \frac{1}{n_m^2}\right]$.

3) if $e\geq0$, $g(t_{g})<0$, there exists two ZCPs, then  $N_\Omega\in\left(0,\frac{1}{n_m^2}\right]\cup\left[\frac{1}{n_{m-1}^2}, \infty\right)$.

4) if $g(t_{l})>0$, $e,g(t_{g})<0$, there exists ZCPs, then $N_\Omega\in\left(0,\frac{1}{n_m^2}\right]\cup\left[\frac{1}{n_{m-1}^2}, \frac{1}{n_{m-2}^2}\right]$.

The proof is completed by summarizing the above conclusions.

\section{ Proof for Lemma \ref{lemma}}
\label{appendix0}
For the high SNR regime, the complementary outage probability (CPOP) of $\tilde C_s$ is defined as
\begin{equation}
\begin{split}
p_c^{\tilde C_s}=Pr\left(\tilde C_s>\tau\right)=Pr\left(\log_2\frac{\gamma_b}{\gamma_e}>\tau\right)=Pr\left(\gamma_b>\mathrm{2}^\tau\gamma_e\right).
\end{split}
\end{equation}
According to the PDFs of $\gamma_b$ and $\gamma_e$, the CPOP of $\tilde C_s$ is
\begin{equation}
\begin{split}
p_c^{\tilde C_s}&=Pr\left(x>\mathrm{2}^\tau y\right)
=\int_0^{+\infty}\frac{1}{\overline\gamma}\mathrm{e}^{-\frac{y}{\overline\gamma}}\left[\int_{\mathrm{2}^\tau y}^{+\infty}\frac{1}{\overline\gamma^K\Gamma(K)}x^{K-1}\mathrm{e}^{-\frac{x}{\overline\gamma}}dx\right]dy\\
&\overset{(a)}{=}\int_0^{+\infty}\mathrm{e}^{-\frac{1+\mathrm{2}^\tau}{\overline\gamma}y}\sum \limits_{m=0} ^{K-1}\frac{(\mathrm{2}^\tau y)^m}{m!\overline\gamma^{m+1}}dy\overset{(b)}{=}\sum \limits_{m=0} ^{K-1}\frac{\mathrm{2}^{\tau m}}{(\mathrm{2}^\tau+1)^{m+1}}=1-\left(\frac{\mathrm{2}^\tau}{\mathrm{2}^\tau+1}\right)^K,
\end{split}
\end{equation}
where ($a$) and ($b$) hold for the integration formula \cite{I.S.Gradshteyn2007}, 3.351.2] and \cite{I.S.Gradshteyn2007}, 3.326.2]. Based on (\ref{equationR}), the CDF of $\tilde R_s$ can be represented as
\begin{align}
F_{\tilde R_s}\left(r\right)= Pr(\tilde C_s-\upsilon\leq r)=1-p_c^{\tilde C_s}\left(r+\upsilon\right)=\left(\frac{\mathrm{2}^{r+\upsilon}}{\mathrm{2}^{r+\upsilon}+1}\right)^K.
\end{align}

\section{ Proof for Theorem \ref{theorem2}}
\label{prooftheorem1}
The partial derivative of the effective throughput $T$ in (\ref{equation20}) w.r.t. $N$ can be deduced as follows
\begin{equation}
\begin{split}
\frac{\partial T}{\partial N}=\frac{B}{N^2}\Xi(N).
\label{derivativeT1}
\end{split}
\end{equation}
Deduce the partial derivative of $\Xi\left(N\right)$ w.r.t. $N$ as
\begin{align}
\frac{\partial \Xi\left(N\right)}{\partial N}
&=-\frac{(H(N))^K}{N^2\left(2^{h(N)}+1\right)}\Xi_1(N),
\end{align}
\begin{align}
\text{with}\qquad   \Xi_1(N)\triangleq\frac{\lambda(N)^2\left(K-2^{h(N)}\right)}{KN\left(2^{h(N)}+1\right)}+2\lambda(N)-\frac{\ln2}{4}Kt\sqrt{N}.
\label{Xi_1}
\end{align}
Obviously, $2\lambda(N)-\frac{\ln2}{4}Kt\sqrt{N}>\frac{3}{2}\lambda(N)>0$ as $\frac{1}{2}\lambda\left(N\right)=\frac{\ln2}{2}KB+\frac{\ln2}{4}Kt\sqrt{N}>\frac{\ln2}{4}Kt\sqrt{N}$, so the sign of $\Xi_1(N)$ is determined by $\left(K-2^{h(N)}\right)$. Specifically, $\Xi_1(N)>0$ holds when $K>2^{h(N)}$, i.e. $N>h^{-1}(\log_2K)$ (as $h(N)$ monotonically decrease with $N$). As for the case $K\leq2^{h(N)}$, the partial derivative of $\Xi_1(N)$ is
\begin{align}
\frac{\partial \Xi_1(N)}{\partial N}
=\frac{\lambda(N)}{N\left(2^{h(N)}+1\right)}\left(\frac{\ln2B}{N}\left(2^{h(N)}-K\right)+\frac{(K+1)\lambda(N)^22^{h(N)}}{K^2N^2\left(2^{h(N)}+1\right)}\right)+\frac{3\ln2Kt}{8\sqrt{N}}>0.\label{dXi_1}
\end{align}
Thus $\Xi_1(N)$ increases monotonically with $N$. Suppose $N_1\triangleq h^{-1}(\log_2K)$,
%based on (\ref{Xi_1}),
we have
\begin{align}
\Xi_1(1)&=\frac{\lambda\left(1\right)^2}{1+2^{B+t}}\left(1-\frac{2^{B+t}}{K}\right)+\ln2K\left(2B+\frac{3t}{4}\right)\nonumber\\
&\overset{(a)}{\approx}-\frac{\lambda\left(1\right)^2}{K}+\ln2K\left(2B+\frac{3t}{4}\right)\leq\left(2-\ln2\left(B+\frac{t}{2}\right)\right)\lambda(1)<0,\\
\Xi_1(N_1)&=\ln2K\left(2B+\frac{3t\sqrt{N_1}}{4}\right)>0,
\end{align}
where ($a$) holds as $\mathrm{2}^{B+t}\gg1$ for the reason that the value of $B$ is relatively large and set $B>50$, without loss of generality. What's more, when $N>N_1$, $2^{h(N)}<K$, we have $\Xi_1(N)>0$ with $N>N_1$ based on the discussion mentioned above. In summary, $\Xi_1(N)$ is negative first and then positive with $N\geq1$. Accordingly, $\Xi(N)$ increases as $N$ increases initially and decreases afterwards.
%$\Xi\left(N\right)=\left(\frac{K\ln2\left(B+\frac{1}{2}t\sqrt{N}\right)}{N\left(1+2^{\frac{B}{N}+\frac{t}{\sqrt{N}}}\right)}+1\right)\left(\frac{\mathrm{2}^{\frac{B}{N}+\frac{t}{\sqrt{N}}}}{\mathrm{2}^{\frac{B}{N}+\frac{t}{\sqrt{N}}}+1}\right)^K-1$.%
%where $\lambda\left(N\right)=\ln2 K\left(B+\frac{1}{2}t\sqrt{N}\right)>0$, $(c)$ holds for $\frac{1}{2}\lambda\left(N\right)=\frac{\ln2}{2}KB+\frac{\ln2}{4}Kt\sqrt{N}>\frac{\ln2}{4}Kt\sqrt{N}$. Obviously, $\Xi\left(N\right)$ decrease as $N$ increase. %
Further, when $N$ approaches $1$ and infinity, we have
\begin{align}
&\Xi\left(1\right)=\left(\frac{K\ln2\left(B+\frac{1}{2}t\right)}{1+2^{B+t}}+1\right)\left(\frac{\mathrm{2}^{B+t}}{\mathrm{2}^{B+t}+1}\right)^K-1\approx\frac{K\ln2\left(B+\frac{1}{2}t\right)}{1+2^{B+t}}>0,
\label{Xi1}\\
&\lim_{N \to +\infty}\Xi\left(N\right)=-1<0.
\label{Xi2}
\end{align}

So effective throughput $T$ increases first and then decreases with $N$, i.e. $T$ is a quasi-concave function of the relaxed $N$. The maximum of $T$ will meet at $\Xi(N^{*})=0$.
\section{ Proof for Corollary 1}
\label{appendix4}
According to the derivative rule for implicit functions, the impact of any parameter $\chi$ on $N$  with $\Xi(N) = 0$ can be analyzed through
\begin{align}
\frac{dN}{d\chi}=-\frac{\partial \Xi/\partial \chi}{\partial \Xi/\partial N}.
\label{equation48}
\end{align}
From Theorem \ref{theorem2}, we already have $\frac{\partial \Xi}{\partial N}<0$ with $N=N^*$. Therefore, the sign of $\frac{dN}{d\chi}$ only depends on $\frac{\partial \Xi}{\partial \chi}$.
The changing trend of $N$ with main system parameters can be discussed as follows:

1) $B$: The partial derivative of $\Xi$ w.r.t. $B$ is
\begin{align}
\frac{\partial \Xi}{\partial B}=\frac{\ln2H(N)^K}{N\left(2^{h(N)}+1\right)}\left(2K+\frac{\lambda(N)\left(K-2^{h(N)}\right)}{N\left(2^{h(N)}+1\right)}\right),
\label{equation51}
\end{align}
where $2K+\frac{\lambda(N)\left(K-2^{h(N)}\right)}{N\left(2^{h(N)}+1\right)}=\frac{K\left(\Xi_1(N)+\frac{\ln2}{4}Kt\sqrt{N}\right)}{\lambda(N)}>0$ as $\Xi_1(N)>0$ when $N=N^*$. In consequence, $\frac{dN}{dB}>0$ at $N=N^*$, i.e. the optimal blocklength $N^*$ increase with $B$.

2) $K$: The partial derivative of $\Xi(N)$ w.r.t. $K$ is
\begin{align}
\frac{\partial \Xi}{\partial K}&=\frac{\lambda(N)H(N)^K}{N\left(2^{h(N)}+1\right)}\left(\frac{1}{K}+\ln H(N)\right)+\ln H(N)H(N)^K\nonumber\\&\overset{(a)}{=}\frac{\left(1-H(N)^K\right)+\ln(H(N)^K)}{K},
\label{equation52}
\end{align}
where ($a$) holds for $\left(\frac{\lambda(N)}{N\left(2^{h(N)}+1\right)}+1\right)\left(H(N)\right)^K=1$ as $\Xi(N^*)=0$. It can be easily deduced that $\frac{\partial \Xi}{\partial K}$ increases with $H(N)^K\in\left(0,1\right)$, result in $\frac{\partial \Xi}{\partial K}<0$.

\section{ Proof for Corollary 2}
\label{appendix5}
According to the chain rule of derivative, the impact of any parameter $\chi$ on $T^*\triangleq T\left(N^*,\chi)\right)$ can be analyzed through
\begin{align}
\frac{dT^*}{d\chi}=\frac{\partial T}{\partial N}\bigg|_{N=N^*}\times \frac{dN^*}{d\chi}+\frac{\partial T}{\partial \chi}\bigg|_{N=N^*}=\frac{\partial T}{\partial \chi}\bigg|_{N=N^*},
\label{equation55}
\end{align}
where the second equality follows from the fact that $\frac{\partial T}{\partial N}\bigg|_{N=N^*}=0$ according to Theorem \ref{theorem2}. Thus the impact of the system parameters on $T(N^*)$ is then discussed as follows:

1) $B$:
\begin{align}
\frac{\partial T}{\partial B}=-\frac{1}{N}\left(\left(\frac{\ln2KB}{N\left(2^{h(N)}+1\right)}+1\right)H(N)^K-1\right),
\label{equation56}
\end{align}
if $N=N^*$, then we can know $\frac{\partial T}{\partial B}>-\frac{1}{N}\Xi(N^*)=0$ from (\ref{equation24}).

2) $K$:
\begin{align}
\frac{\partial T}{\partial K}=-\frac{B}{N}H(N)^K\ln H(N)>0.
\label{equationTK}
\end{align}

3) $\bar\epsilon$ or $\bar\delta$:
the partial derivative of $T$ about $t$ can be represented as
\begin{align}
\frac{\partial T}{\partial t}=-\frac{B}{N}\frac{\ln2KH(N)^K}{\sqrt{N}\left(2^{h(N)}+1\right)}<0,
\end{align}
thus we will obtain the conclusion that $T^*$ increases with $\bar\epsilon$, $\bar\delta$ as both $\bar\epsilon$ and $\bar\delta$ decrease with $t$.
\end{appendices}

\renewcommand{\baselinestretch}{1.1}

\end{document}